%% file: pebbles.tex
\documentclass[sigconf]{acmart}

\makeatletter
\def\genfootnote{\gdef\@thefnmark{}\@footnotetext}
\makeatother

\usepackage{balance}

\usepackage{xcolor}
\definecolor{midgray}{gray}{0.6}
\usepackage{tikz}
\usetikzlibrary{arrows}

\newtheorem{theorem}{Theorem}
\newtheorem{lemma}[theorem]{Lemma}

\usepackage{makecell}

\begin{document}

\title{On the Hardness of Red-Blue Pebble Games}

\author{P\'al Andr\'as Papp}
\affiliation{
  \institution{ETH Zurich}
}
\email{apapp@ethz.ch}

\author{Roger Wattenhofer}
\affiliation{
  \institution{ETH Zurich}
}
\email{wattenhofer@ethz.ch}

\begin{abstract}
Red-blue pebble games model the computation cost of a two-level memory hierarchy. We present various hardness results in different red-blue pebbling variants, with a focus on the \textsc{oneshot} model. We first study the relationship between previously introduced red-blue pebble models (\textsc{base}, \textsc{oneshot}, \textsc{nodel}). We also analyze a new variant (\textsc{compcost}) to obtain a more realistic model of computation. We then prove that red-blue pebbling is NP-hard in all of these model variants. Furthermore, we show that in the \textsc{oneshot} model, a $\delta$-approximation algorithm for $\delta<2$ is only possible if the unique games conjecture is false. Finally, we show that greedy algorithms are not good candidates for approximation, since they can return significantly worse solutions than the optimum.
\end{abstract}

\ccsdesc[300]{Theory of computation~Problems, reductions and completeness}
\ccsdesc[300]{Theory of computation~Models of computation}
\ccsdesc[100]{Computer systems organization~Processors and memory architectures}

\keywords{Red-blue pebble game; Time-memory trade-off}

\maketitle

\section{Introduction} \label{sec:Intro} \genfootnote{The short version of the paper is published in the \textit{Proceedings of the 32nd ACM Symposium on Parallelism in Algorithms and Architectures (SPAA 2020)}, with ACM DOI 10.1145/3350755.3400278.}

Pebble games allow us to model different aspects of a computation process. In particular, red-blue pebble games model the I/O cost of a computation on a two-level memory hierarchy. While machines usually have sufficient slow memory (e.g. RAM) to store all intermediate values of a computation, fast-access memory (e.g. cache) in machines is often heavily limited. However, the red-blue pebble game can also model other fast/slow memory combinations, e.g. main memory vs. disk, or different levels of cache.

Since fast memory is limited, a machine could be forced to save some intermediate values of a complex calculation in slow memory, and later retrieve these values when needed. However, moving values between fast and slow memory is costly; in fact, this is often the bottleneck of a computation. Hence minimizing the required number of transfers between fast and slow memory is a crucial problem in many application areas, with High-Performance Computing being a prominent example \cite{hoefler}. As such, red-blue pebble games are a perfect example for a playful theoretical question that is also highly relevant in practice.

Any computation can be modeled as a directed acyclic graph (DAG), with source nodes corresponding to the inputs of the computation, and sinks corresponding to the outputs. Intermediate values are the nodes of the graph, with the input edges of a node $v$ specifying the set of values that are required for the computation of $v$. Nodes that we have already computed and saved in memory are marked with a pebble. In the red-blue pebbling game, we have two kinds of available pebbles: red pebbles correspond to values saved in fast memory, while blue pebbles denote values saved in slow memory. To compute a specific node $v$ of the DAG, we need to have all of its inputs available in fast memory, i.e. have a red pebble on all the nodes that have an edge to $v$.

When pebbling a DAG, we initially assume that none of the nodes contain a pebble. Then computation happens through repeatedly applying the following steps:

\vspace{3pt}
\begin{enumerate}
    \setlength{\itemsep}{6pt}
	\item {\color{midgray} \hrule height 1.3pt} \vspace{7pt}
	\texttt{Move to fast memory:} replace a blue pebble by a red pebble at any node.
	\item \texttt{Move to slow memory:} replace a red pebble by a blue pebble at any node.
	\item \texttt{Compute:} if all input nodes of a node $v$ already have a red pebble, then place a red pebble on node $v$.
	\item \texttt{Delete:} remove a (blue or red) pebble from any node. \vspace{7pt}
	{\color{midgray} \hrule height 1.3pt}
\end{enumerate}
\vspace{3pt}

A \textit{pebbling} is a sequence of these steps where in the final state, each sink node of the DAG is computed, i.e. it has a (blue or red) pebble on it. We interpret Step 3 to always allow placing a red pebble on a source node, since it has no inputs at all; thus a pebbling can always begin by applying Step 3 on the source nodes of the DAG.

The \textit{cost} of a pebbling is defined as the total number of transfer operations (Steps 1 and 2) executed throughout the process. Steps 3 and 4 are considered free in this basic setting, but they do incur some cost in more refined model variants.

The aim of the pebbling problem is to model the time-memory tradeoff, i.e. to help us understand the cost of a computation when fast memory is a limited resource. To achieve this, we consider a parameter $R$, and we limit the legal pebblings such that at any point in the process, there are \textit{at most} $R$ red pebbles placed on the nodes of the DAG. The pebbling of a DAG may be cheap for a larger $R$ value, but it can become significantly more costly if the number of available red pebbles is reduced, so this indeed allows us to study a time-memory tradeoff.

Given a specific number of red pebbles $R$, the goal of pebbling is to pebble the graph at the lowest possible cost $C$. Formally, the decision version of the \textsc{Pebbling} problem is as follows: given an input DAG and integers $R$ and $C$, does there exist a pebbling of the DAG with $R$ red pebbles and cost at most $C$?

We refer to the above defined setting as the \textsc{base} version of the problem. In the related literature, other versions of the problem have also been introduced, in order to make the problem computationally more feasible and/or theoretically more interesting; however, the connection between these variants remained unclear. We begin with an overview of the different model variants in Section \ref{sec:versions}, and an analysis of their relation to each other. Furthermore, we also study \textsc{compcost}, a novel red-blue pebbling variant, and we argue that \textsc{compcost} is a more realistic model. We also discuss how the time-memory tradeoff can behave in general in these models.

We then present a proof of NP-hardness for all of these pebbling variants, through a reduction from the Hamiltonian Path problem. Note that NP-hardness in some of the models was already known before \cite{RBcomplex}; however, our new reduction provides a notably simpler proof of this result.

As pebbling is an NP-hard problem with relevant applications, it is a crucial question if the problem can be tackled with approximation algorithms or heuristic methods in practice. Our main contributions are novel insights into the limits of these approaches. In Section \ref{sec:VC}, through a reduction from Vertex Cover, we show that \textsc{oneshot} red-blue pebbling cannot be approximated to any factor less than 2 if the unique games conjecture holds. Finally, in Section \ref{sec:greedy}, we analyze the natural greedy strategies for the problem, and show that the solutions they return can be significantly worse than the optimum cost; as such, we cannot expect these algorithms to guarantee a good approximation ratio.

\section{Related Work}

The introduction of graph pebbling problems dates back to the 1970s, with the most studied variants being the standard pebble game (modeling computational time-memory tradeoffs in general), and black-white pebbling (which also introduces non-determinism into this model). Among the earliest results on the topic is the PSPACE-completeness of standard pebbling \cite{Bcomplex}, a wide range of results on the time-memory tradeoff in these pebbling games \cite{TvsSP1, TvsSP3, TvsSP4}, and a characterization of the difference between the two settings \cite{BvsBW1, BvsBW2}.
Some further complexity and inapproximability results on these games are available in \cite{Badditive}, \cite{BNPcomp} or \cite{Inapprox}.

\begin{figure*}
\hspace{0.01\textwidth}
\minipage{0.46\textwidth}
	\input{constdeg.tikz}
	\vspace{-13pt}
	\caption{Transformation for $\Delta = O(1)$: replacing the target node of an input group by $h$ distinct layers}
	\label{fig:constdeg}
\endminipage\hfill
\hspace{0.06\textwidth}
\minipage{0.46\textwidth}
	\centering
	\input{h2c.tikz}
	\vspace{1pt}
	\caption{The hard-to-compute (H2C) gadget}
	\label{fig:h2c}
\endminipage\hfill
\hspace{0.01\textwidth}
\end{figure*}

In contrast to standard pebbling, red-blue pebbling allows us to save previously computed values (by using blue pebbles), and hence it allows to model the I/O complexity of a computation. 
Red-blue pebble games were also introduced in the 1980s by Hong \textit{et. al.} \cite{RBintro}, who mostly discussed the resulting time-memory tradeoff for some specific computational tasks.
A thorough investigation of the red-blue pebble game has recently been conducted by Demaine and Liu in \cite{RBcomplex}. Their study shows that the \textsc{base} version of red-blue pebble games is PSPACE-complete, through a reduction to the standard (black) pebbling game. They also show that without deletions, the problem is NP-complete and  W[1]-hard in the maximal cost allowed.

Besides this, the result of Carpenter \textit{et. al.} studies red-blue pebbling in the \textsc{oneshot} model, presenting an approximation algorithm with a cost of at most $O(\texttt{opt}(R) \cdot \log^{3/2}n)$ and the use of at most $O(R \cdot \log^{3/2}n)$ red pebbles \cite{RBapprox}, where $\texttt{opt}(R)$ denotes the optimal pebbling cost of the DAG with $R$ red pebbles. The same paper also discusses a natural generalization of pebbling to multi-level memory hierarchies, i.e. hierarchies with more than 2 levels.

The work of \cite{compcost} reduces the study of bandwidth-hard functions to red-blue pebbling problems, and introduces a model that is similar to our \textsc{compcost} variant. However, in this new model, the authors only show that the problem still remains NP-hard.

Close-to-optimal pebbling strategies on special classes of graphs, motivated by financial applications, are described in \cite{RBpyramids1, RBpyramids2}. 
In \cite{RBdecompose}, the authors analyze a decomposition technique on DAGs that allows to derive lower bounds on the pebbling cost of specific computational tasks in the \textsc{oneshot} model.
Finally, the work of \cite{RBparallel} discusses a generalized version of red-blue pebbling (with multiple `shades' of red pebbles) in order to model a parallel execution on multiple processors. Furthermore, \cite{RBparallel} also presents some lower bounds for the pebbling cost of well-known numerical algorithms.

\section{Basic properties and observations} \label{sec:basics}

\paragraph*{Straightforward bounds} We first discuss some basic properties of pebbling. We use $n$ to denote the number of nodes in the DAG, and $\Delta$ to denote the largest indegree.

First, note that a computation of a node $v$ with indegree $d_{\text{in}}$ requires $d_{\text{in}}+1$ red pebbles: 1 pebble at $v$, and $d_{\text{in}}$ more at $v$’s input nodes. Therefore if $R < \Delta + 1$, a pebbling is not possible at all. Conversely, if $R \geq \Delta + 1$, a valid pebbling always exists: we can simply take a topological ordering of the DAG, and always compute the next vertex $v$ of the ordering (through Step 3) by moving the red pebbles to the inputs of $v$, and changing all red pebbles to blue in every other node of the DAG. Hence in the following, we always assume $R \geq \Delta + 1$.

Furthermore, the optimal pebbling cost of any DAG is at most $(2 \Delta+1) \cdot n$ \cite{RBcomplex}. Following the greedy strategy outlined above, the computation of each new node requires at most $\Delta+1$ instances of Step 2 (making $\Delta+1$ red pebbles available), and then $\Delta$ instances of Step 1 (making the inputs of the next computable node red). This sums up to a cost of $(2 \Delta+1) \cdot n$ over the whole process.

\paragraph*{Constant indegree} The main idea behind most of our DAG constructions is to have specific node groups (so-called \textit{input groups}) of size $R-1$, which are all inputs of a specific \textit{target node} $t$. This simplifies the analysis of pebbling significantly, since each such target $t$ can only be computed by using all the available red pebbles; thus we do not have to discuss which red pebbles to move to the inputs of $t$, or where the leftover red pebbles are in the DAG. An optimal pebbling strategy in such DAGs comes down to the order in which we \textit{visit} these input groups, i.e. the order of computing the target nodes.

However, this approach often requires large input groups, and hence a large $\Delta$. This is in heavy contrast with most application areas, where it is reasonable to assume that $\Delta = 2$ or $3$. Previous results on red-blue pebbling often also assume that $\Delta = O(1)$ \cite{RBcomplex}. Hence, we also describe a transformation technique which shows that each of our results also hold for DAGs with constant indegree.

The main idea of this technique is to use the gadget shown in Figure \ref{fig:constdeg}, and increase the number of available red pebbles to $R'=R+1$. In order to compute all nodes in this gadget, all the $R-1$ nodes on the left side are repeatedly needed. Besides this, we always need 2 red pebbles in the layers of the gadget to compute the next node. Hence, if we were to use less than $R+1$ red pebbles to pebble the gadget, then red pebbles would have to be moved around among the $R-1$ left-side nodes repeatedly, summing up to a cost which is proportional to the length $h$. If $h$ is high enough, then this results in an unreasonably high cost for pebbling the gadget.

Thus the gadget has indegree 2, but it achieves the same goal as our original input group: it practically forces us to place $R-1$ red pebbles in the left side of the gadget at some point, in any reasonable pebbling. In Appendix \ref{App:B}, we describe this transformation in more detail, and discuss how to adapt each of our constructions to the $\Delta = O(1)$ case.

Previous works have often used a pyramid gadget \textit{pyramid gadget} \cite{Bcomplex, RBcomplex, RBpyramids1} to reduce the indegree; this also requires a high number of red pebbles to pebble with minimal cost, while still having $\Delta = O(1)$. However, if the number of red pebbles is reduced by $1$ for a pyramid, the increase in cost is only $2$, whereas in our gadget, taking a single red pebble away already increases the cost dramatically. This is a key property for a simple analysis in our proofs.

\paragraph*{Computing the source nodes} We also have to discuss the source nodes of the DAG separately, since the reasonable way to model these nodes might vary by application area. In some practical computations, there are numerous input values, which thus naturally have to be stored in slow memory initially, and loading them into fast memory adds an inherent cost to the computation. In other cases, the computation inputs (i.e. source nodes) might be trivial values that can be calculated at practically no cost.

Our definition of the \textsc{base} model describes the second setting: we assume that a red pebble can be placed on the source nodes freely at any point. However, with a minor modification, we can also model the first setting in our base model. We achieve this by adding a so-called \textit{hard-to-compute} (H2C) \textit{gadget} in front of each source node $v$ of the DAG. The gadget is shown in Figure \ref{fig:h2c}.

The main idea of the H2C gadget is that all $R$ red pebbles are required to compute any of the \textit{starter nodes} $u_1$, $u_2$ or $u_3$. Hence when computing the last of these $3$ nodes, the previous $2$ must already be computed and turned to blue. These $2$ nodes then have to be loaded back from slow memory in order to compute $v$. This implies that computing $v$ indirectly requires at least 4 transfer operations, and thus it now has a constant cost of 4 (the fact that this cost is 4 instead of only 1 does not matter asymptotically).

In order to ensure this inherent cost for every source node, we do not need an entirely separate copy of the H2C gadget for each source node. Instead, the node $s$ and the group $B$ of $R-1$ nodes can be common in all the distinct H2C gadgets added for the distinct source nodes, and we only need to instantiate the three starter nodes $u_1$, $u_2$, $u_3$ separately for each source $v$. This way, we add 3 extra nodes for every source of the DAG, and a further $R$ extra nodes to the DAG altogether. This does not change the magnitude of the number of nodes, and thus it also does not affect any of our constructions in the paper.

\setlength\tabcolsep{4pt}
\begin{table*}
\centering
  \begin{tabular}{ c || c | c | c | c || c }
    \hline
    Model& \small \makecell{Blue \\ to red}   & \makecell{Red to \\ blue} & Compute & Delete & Description \\ \hline \hline
    \textsc{base} & 1 & 1 & 0 & 0 & Baseline model (see Section \ref{sec:Intro}) \\ \hline
		\textsc{oneshot} & 1 & 1 & $0,\infty,\infty$, ...  & 0 & Each node only computable once \\ \hline
		\textsc{nodel} & 1 & 1 & 0 & $\infty$ & Pebbles cannot be deleted \\ \hline
		\textsc{compcost} & 1 & 1 & $\epsilon$ & 0 & Computation also has a cost of $\epsilon$ \\
		\hline
  \end{tabular}
  \vspace{6pt}
    \caption{Summary of the cost of operations in different models.}
	\label{tab:model}
\end{table*}

\setlength\tabcolsep{4pt}
\begin{table*}
\centering
  \begin{tabular}{ c || c | c | c | c }
    \hline
    Model& \small \makecell{Cost of \\ optimal pebbling } & \makecell{Length of \\ optimal pebbling} & Complexity & \makecell{Ratio of greedy \\ to optimum} \\ \hline \hline
    \textsc{base} & $\in [0, \, (2 \Delta \! + \! 1) \! \cdot \! n]$ & Up to $\omega(\text{poly}(n))$ & \small PSPACE-complete \cite{RBcomplex} & $\Omega(n^{1/6})$ \\ \hline
		\textsc{oneshot} & $\in [0, \, (2 \Delta \! + \! 1) \! \cdot \! n]$ & $O(\Delta \! \cdot \! n)$ & \small NP-complete & $\widetilde{\Omega}(\sqrt{n})$ \\ \hline
		\textsc{nodel} & $\in [n, \, (2 \Delta \! + \! 1) \! \cdot \! n]$ \small\cite{RBcomplex} & $O(\Delta \! \cdot \! n)$ & \small NP-complete \cite{RBcomplex} & Large $\Theta(1)$ \\ \hline
		\textsc{compcost} & $ \in [\epsilon \! \cdot \! n, \, (2 \Delta \! + \! 1 \! + \! \epsilon) \! \cdot \! n] $ & $O(\Delta \! \cdot \! n)$ & \small NP-complete & Large $\Theta(1)$  \\
		\hline
  \end{tabular}
  \vspace{7pt}
	\caption{Summary of the basic properties of different models, and our results. The inapproximability result is not shown in the table, because it only applies to the \textsc{oneshot} model. Recall that the lower bound on the optimum cost in \textsc{nodel} and \textsc{compcost} are both only asymptotic, assuming that $R$ (in \textsc{nodel}) or the number of source nodes (in \textsc{compcost}) is in $o(n)$.}
	\label{tab:model2}
\end{table*}

\paragraph*{Disabling the recomputation of nodes}

Besides modeling an inherent cost for source nodes, there is another important application of the H2C gadget. In particular, we widely use the gadget in our constructions to ensure that specific nodes are costly to recompute if they are ever deleted during a pebbling.

Consider the H2C gadget after the node $v$ has been computed. Once the red pebbles are moved away from the starter nodes of $v$, turning the starters red again costs 3, while recomputing them from scratch costs at least 4. In contrast to this, if we simply turn $v$ blue after it has been computed, and then red again when needed, this only has a cost of 2. Thus having computed $v$ once, a reasonable pebbling will never delete the pebble from $v$ (until its very last use), but always save $v$ to slow memory and transfer it back later instead. Therefore, the H2C gadget can also be used to indirectly ensure that some nodes are never deleted and then recomputed later, but always saved with a blue pebble instead.

\paragraph*{Small number of source nodes} In some applications, we might also want to restrict ourselves to DAGs with constantly many source nodes; we point out that our results also hold under this restriction. Given any DAG construction, we can easily adapt it to this setting by adding a single new source $s_0$, drawing an edge from $s_0$ to every other node of the DAG, and increasing the number of available red pebbles to $R'=R+1$. Since $s_0$ is now required for every computation, a reasonable pebbling never removes a red pebble from $s_0$, which leaves $R$ red pebbles to pebble the rest of the DAG as in the original case. Thus the addition of $s_0$ results in a DAG with essentially the same behavior as before, but only a single source node.

\paragraph*{Initial and final state of a pebbling} Different papers on the topic have slightly different definitions for the starting or finishing state of pebblings. For example, some papers assume that source nodes of the DAG have to be computed explicitly (as in our case), while others assume an initial blue pebble on sources. Similarly, some require a pebble of any color on the sink nodes in the finishing state (as in our definition), while others explicitly require a blue pebble on sinks. With a few simple observations, one can show that these different variants of the problem definition are essentially equivalent for our purposes, and thus our results also hold in these slightly different settings. We briefly discuss this in Appendix \ref{App:C}.

\section{Models of red-blue pebble games} \label{sec:versions}

While the \textsc{base} variant of the red-blue pebble game was the first to be defined, most of the related work studies different versions of the game. The \textsc{base} version allows us to use Steps 3 and 4 any number of times at no cost, and hence it might compute nodes through a very long sequence of deletion and recomputation steps in the optimal strategy. Because of this, it is possible in the \textsc{base} version that any optimal pebbling of a DAG consists of superpolynomially many steps, but still has low (possibly zero) cost.

Regarding practice, this is unrealistic, since computation clearly has nonzero cost. Regarding theory, it places the problem outside of NP (unless NP=PSPACE) \cite{RBcomplex}, since any optimal sequence of steps might be too long to verify; this makes the problem undesirably hard. Thus, related work has often diverted from the \textsc{base} version, introducing minor changes to ensure that (i) the model makes more sense practically, and/or (ii) the problem is actually in NP.

One such variant is the \textsc{oneshot} red-blue pebble game (also known as red-blue-white pebbling), where Step 3 can be executed on each node at most once throughout the pebbling. This rule directly forbids recomputation. As we will see in Lemma \ref{lem:lengthbound}, this already ensures that any optimal pebbling consists of $O(\Delta \cdot n)$ steps, so the problem is in NP.

Another variant considered in \cite{RBcomplex} is the red-blue pebble game with no deletions (\textsc{nodel}), where Step 4 is not available. Step 3 still allows us to replace a blue pebble by a red one if all inputs contain a red pebble, so this variant of the game does allow recomputation. However, we cannot simply delete a red pebble when it is not needed, and recompute it later at no cost, as in \textsc{base}; instead, we have to replace it by a blue pebble temporarily (using Step 2), which incurs a cost of 1. Hence, the model ensures that the recomputation of a node has an indirect cost, disallowing the exponentially long but cost-free deletion-recomputation sequences in the optimal pebbling. It was already shown in \cite{RBcomplex} that \textsc{nodel} pebbling is NP-complete, through a reduction from Positive 3-in-1 SAT.

In \textsc{nodel}, the fact that pebbles cannot be deleted from nodes means that at least $n-R$ nodes have to become blue by the end of the process, and thus the cost of any pebbling is at least $n-R$. As $R$ is usually significantly smaller than $n$ (we can assume $R=o(n)$), this shows that the minimal cost of a pebbling is in the magnitude of $n$. Since the maximal cost of any pebbling is $(2 \Delta + 1) \cdot n$ (as shown in Section \ref{sec:basics}), this means that the cost difference between any two pebbling strategies differs by at most a small $(2 \Delta + 1)$ factor in \textsc{nodel}. One the other hand, the \textsc{oneshot} model might allow the cost of a pebbling to go down to as low as 0, so strategies can generally differ by a much larger factor. This makes the \textsc{oneshot} model the most interesting model variant from a theoretical perspective.

However, a red-blue pebbling variant should not only be theoretically interesting, but also practically useful. The \textsc{base} model does not model practice well, as it allows exponentially long chains of free recomputations. Both the \textsc{oneshot} and the \textsc{nodel} model aim to provide a more realistic model of pebbling. Step 3 being free is somewhat motivated, as computing a value (from inputs in cache) is usually much faster than moving a value between fast and slow memory (Steps 1 or 2). Thus in some cases, using recomputation steps may indeed be the most efficient way to execute a sequence of computations. Since recomputation is not allowed in the \textsc{oneshot} model, the \textsc{oneshot} model does not allow as much freedom as one would hope. Similarly, the \textsc{nodel} model forces us to save every intermediate value into slow memory instead of simply allowing to delete it, so it also restricts our action space unnecessarily.

Consequently, we also discuss a new red-blue pebbling variant, which we believe to be more realistic than the previous ones. In the \textsc{compcost} model, the setting is the same as in the \textsc{base} model, with the difference that Step 3 is not free, but has a cost of $\epsilon$ for some small constant $0 < \epsilon < 1$. Note that this is in line with the actual behavior we want to model: while computational steps (Step 3) are not nearly as costly as transfer operations (Steps 1 and 2), they do incur some minimal cost. In reality, the cache is roughly 100 times faster than a bus access, so $\epsilon \approx 1/100$.

We note that a very similar setting has already been introduced in the work of \cite{compcost}, with the additional restriction that $\epsilon \geq \frac{1}{3 n}$ (which is asymptotically equivalent to our case of accepting any $\epsilon$). The work of \cite{compcost} shows that pebbling still remains NP-hard in this modified model. However, we find that a much more interesting property of this setting is that it actually makes the problem easier: assigning cost $\epsilon$ to Step 3 is in fact already enough to insert the problem into NP. Intuitively, since computations now also have a cost, the sequence of deletion-recomputations cannot be too long in an optimal pebbling, and thus the length of any optimal pebbling becomes polynomial in $n$. Hence, \textsc{compcost} is not only a more realistic model of computations, but it is also a more natural way to place the problem into NP than either \textsc{oneshot} or \textsc{nodel}.

It remains to formally show that the length of the optimal pebbling is at most $O(\Delta \cdot n)$ in all the modified model variants, and thus the problem falls into NP in these models.

\begin{lemma} \label{lem:lengthbound}
Given a \textsc{Pebbling} problem in the \textsc{oneshot}, \textsc{nodel} or \textsc{compcost} model, any optimal pebbling strategy consists of at most $O(\Delta \cdot n)$ steps.
\end{lemma}

\begin{proof}

Recall from Section \ref{sec:basics} that there always exists a pebbling with cost at most $(2 \Delta + 1) \cdot n$ for any DAG. This upper bound changes to $(2 \Delta + 1 + \epsilon) \cdot n$ in the \textsc{compcost} model due to the extra cost of computations. Since transfer steps (Steps 1 and 2) have a cost of 1, this implies that an optimal pebbling can contain at most $O(\Delta \cdot n)$ transfer steps. It only remains to also bound the number of Steps 3 and 4 in these models.

\begin{figure*}
\minipage{0.47\textwidth}
	\centering
	\vspace{22pt}
	\input{tradeoff.tikz}
	\vspace{8pt}
	\caption{Example DAG for time-memory tradeoff. Edges from the control groups are shown together for simplicity.}
	\label{fig:tradeoff_graph}
\endminipage\hfill
\hspace{0.02\textwidth}
\minipage{0.49\textwidth}
	\centering
	\input{diagram.tikz}
	\vspace{-5pt}
	\caption{Tradeoff diagram for the DAG shown in Figure \ref{fig:tradeoff_graph}.}
	\label{fig:tradeoff_diag}
\endminipage\hfill
\end{figure*}

The \textsc{oneshot} version naturally implies that there are at most $n$ executions of Step 3 in a pebbling. Once a pebble has been deleted from a node, there is no way to place a pebble on this node again, so Step 4 can also be called at most $n$ times. Thus any optimal pebbling indeed consists of $O(\Delta \cdot n)$ steps.

In \textsc{nodel}, Step 4 is not available. Each computation in a pebbling is either a first one at a node (there are at most $n$ of these), or it is a recomputation step that replaces a specific blue pebble. Since blue pebbles are only created in Step 2, which is invoked at most $O(\Delta \cdot n)$ times, the number of recomputations is also within $O(\Delta \cdot n)$.

In the \textsc{compcost} model, assume that the number of Steps 3 and 4 is altogether $p$. Since each deletion step removes a pebble that was previously placed on the DAG, the number of deletions is at most the number of computations. Thus at least half of the non-transfer steps are computations, which means that Steps 3 and 4 altogether have a cost of at least $\epsilon \! \cdot \! \frac{p}{2}$. This shows that if $p > \frac{2}{\epsilon} \cdot (2 \Delta \! + \! 1 \! + \epsilon) \! \cdot \! n$ in a pebbling, then the cost of non-transfer steps is already larger than $(2 \Delta \! + \! 1 \! + \epsilon) \! \cdot \! n$, so the pebbling is not optimal. Thus we have $p \leq \frac{2}{\epsilon} \cdot (2 \Delta \! + \! 1 \! + \epsilon) \! \cdot \! n = O(\Delta \cdot n)$ in any optimal pebbling. \qedhere
\end{proof}

Note that in the \textsc{compcost} model, the minimal cost of a pebbling is in the magnitude of $\epsilon \cdot n$ (unless most nodes of the DAG are source nodes), since each non-source node has to be computed at least once. 
Therefore the minimal and maximal cost are within a constant factor $\frac{2 \Delta + 1 + \epsilon}{\epsilon}$, similarly to the case of \textsc{nodel}. Because of this, the \textsc{oneshot} model still remains the most interesting one theoretically. Hence, when discussing the results of the paper, we primarily focus on the \textsc{oneshot} model, and then go on to discuss the applicability of the specific result to other models.

We illustrate the cost of operations in Table \ref{tab:model}, and summarize the main properties of the models and our results in Table \ref{tab:model2}.

\normalsize

\section{Time-memory tradeoff} \label{sec:tradeoff}

We now present a DAG that allows us to analyze the worst-case tradeoff between the parameter $R$ and the optimal pebbling cost with $R$ red pebbles, denoted by $\texttt{opt}(R)$.

Let us focus on the \textsc{oneshot} model. Recall that the minimal cost of pebbling in this model is 0, and the maximal cost is $(2 \Delta +1) \cdot n$. Furthermore, we always have $\texttt{opt}(R-1) \leq \texttt{opt}(R) + 2n$ (assuming, of course, that $R-1 \geq \Delta+1$). This is because with $R-1$ red pebbles, we can basically follow the same strategy as with $R$ red pebbles, maintaining the invariant that our $R-1$ red pebbles always occupy the same nodes as in the original strategy, with one of the red pebbles missing.
For each computation step when the missing red pebble is among the input nodes, we can select another red pebble that is not used in this computation, and move this red pebble here for the computation at a cost of 2. Since there are at most $n$ computations, the technique adds an extra cost of at most $2n$ to $\texttt{opt}(R)$. In other words, each time we allow a further red pebble, the optimum cost can decrease by at most $2n$.

Consider the DAG construction shown in Figure \ref{fig:tradeoff_graph}, which consists of two control groups of size $d$, and a chain of nodes that are each enabled by the previous node in the chain, and one of the two control groups in an alternating fashion. In case of a long chain, the control groups amount to a negligibly small part of this DAG, so for simplicity, we will now use $n$ to denote the length of the chain. An optimal pebbling strategy of the DAG keeps most of the red pebbles in the control groups, and only has 2 red pebbles in the chain at any time: after computing the next chain node, the previous chain node is immediately deleted, as it is not needed ever again.

In this graph, $\texttt{opt}(2d+2) = 0$: if we can always keep both control groups in cache, then the chain can sequentially be computed at no cost.
On the other hand, $\texttt{opt}(d+2) = 2d \cdot n$: with only $d+2$ red pebbles, we have to transfer all $d$ red pebbles from one control group to the other (at a cost of $2d$) for the computation of every single chain node. With $\Delta = d+1$, this translates to $\texttt{opt}(2 \Delta) = 0$ and $\texttt{opt}(\Delta + 1) = (2 \Delta - 2) \cdot n$, showing that the cost ranges from 0 to almost our upper bound of $(2 \Delta + 1) \cdot n$.

What makes the example more interesting, however, is the fact that the optimum gradually decreases between these points. Starting from $R=2d+2$, whenever we take away a red pebble, it means that for each node of the chain, 2 extra transfer operations are required to move another red pebble to the other control group. I.e., with $d+2+i$ red pebbles available, we have to move $d-i$ red pebbles for computing each chain node, and thus $\texttt{opt}(d+2+i) = 2(d-i) \cdot n$ for every $i \in [0, d]$. This means that it is indeed possible to exhibit the maximal drop $2n$ in every step of the function $\texttt{opt}(R)$ while going from almost the upper bound $(2 \Delta + 1) \cdot n$ down to 0, as illustrated by the tradeoff diagram in Figure \ref{fig:tradeoff_diag}.

We also obtain a similar tradeoff diagram for the remaining three models, after adding a H2C gadget to ensure that nodes in the control groups are never recomputed. In the \textsc{base} model, the diagram is essentially the same as in \textsc{oneshot}.
In \textsc{nodel}, each chain node has to be turned blue instead of being deleted, hence every $\texttt{opt}(R)$ value is increased by $n$. In \textsc{compcost}, each computation costs $\epsilon$, so $\texttt{opt}(R)$ values are increased by $\epsilon \cdot n$. These tradeoff diagrams are discussed in Appendix \ref{App:A:tradeoff}.

\section{NP-hardness} \label{sec:NP}

We now present a simple proof of NP-hardness for the pebbling problem through a reduction from the Hamiltonian Path problem, which is long known to be NP-complete \cite{HampathHard}. Unlike our other results, we first prove this result in the \textsc{nodel} model. The proof for the remaining models is discussed in Appendix \ref{App:A:NP}.

\renewcommand{\proofname}{Proof for the \textit{\textsc{nodel}} model.}
\begin{theorem} \label{th:hampath}
The \textsc{Pebbling} problem is NP-complete in the \textsc{oneshot}, \textsc{nodel} and \textsc{compcost} models, and NP-hard in the \textsc{base} model.
\end{theorem}

\begin{proof}
Lemma \ref{lem:lengthbound} already points out that except for the \textsc{base} model, the problem is within NP, since the optimal pebbling consists of $O(\Delta \cdot n)$ steps.

\begin{figure*}
	\centering
	\input{hampath.tikz}
	\caption{Input groups created for two nodes of $G$ in Theorem \ref{th:hampath} if they are not connected (left) and if they are connected (right). In the second case, the corresponding contact nodes are merged into one.}
	\label{fig:hampath}
\end{figure*}
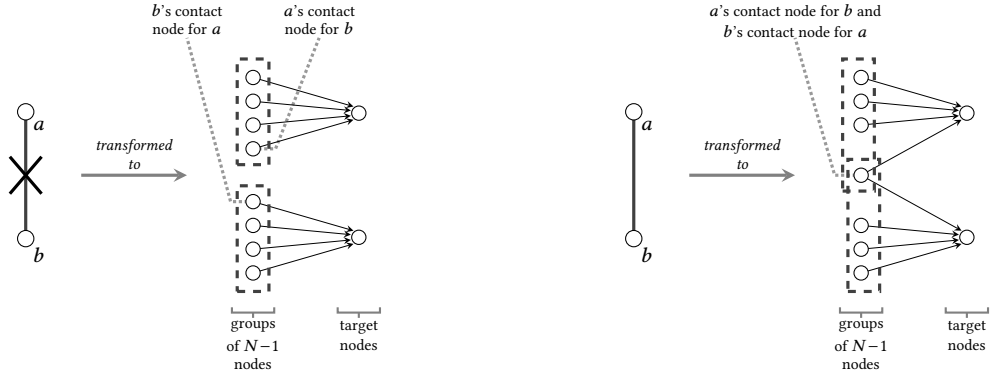

Assume we have a graph $G$ on $N$ nodes and $M$ edges, in which we want to solve the Hamiltonian Path problem; we show how to convert this into a pebbling problem in a DAG. To avoid confusion, we use $a$ and $b$ to denote the nodes of the original graph $G$, while we use $t$, $u$ and $v$ to denote the nodes in the DAG.

For our reduction, we create $N$ distinct \textit{target nodes} that are sinks of the DAG, each representing a distinct node of $G$. For a node $a$ of $G$, we denote the corresponding target node by $t_a$. To each target node $t_a$, we specify a set of exactly $N-1$ distinct nodes that are the inputs of (i.e. have an outgoing edge to) $t_a$; we refer to this set as the input group of node $a$.

These input groups are formed as follows: for every node $a$, let us consider all the other $N-1$ nodes of $G$ (excluding $a$), and for each such other node $b$, we create a specific node $v_{a,b}$ in the DAG. We refer to this node $v_{a,b}$ as the \textit{contact node} in group $a$ for node $b$. We then select the $N-1$ contact nodes in group $a$ as the input group for the target node $t_a$. This provides a DAG with $N \cdot (N-1)$ source and $N$ sink nodes.

Finally, for each edge $(a,b)$ of $G$, we merge the two corresponding contact nodes (i.e. the contact node in group $a$ for node $b$, and the contact node in group $b$ for node $a$) into a single node, as illustrated in Figure \ref{fig:hampath}. That is, if $a$ and $b$ are neighbors in $G$, then the new merged node $v_{a,b}=v_{b,a}$ will be an input of both $t_a$ and $t_b$ in the DAG, but if $a$ and $b$ are not neighbors, then the inputs of $t_a$ and $t_b$ will remain disjoint. This gives us a DAG with altogether $N \cdot (N-1) - M$ source nodes and $N$ sink nodes. We consider the pebbling problem on this graph with $R=N$; note that this is the minimal possible $R$ since $\Delta=N-1$.

Each pebbling of the graph has to visit the input groups in some order to compute all the sink (target) nodes. We need all the $N$ available red pebbles for every such computation, i.e. we need to place $N-1$ red pebbles on the nodes of the input group, and the final red pebble on the newly computed sink. Hence the computations of the target nodes are distinct steps during the pebbling process when we know the position of all the red pebbles; this essentially allows us to characterize the entire pebbling by the order in which the target nodes are computed.

Between the computation of two distinct targets $t_a$ and $t_b$, we always have to (i) turn the previously computed sink blue to free the red pebble from the sink, and (ii) move the red pebbles from one input group to the other to enable the computation. The first step always has a cost of 1. In the second step, moving a red pebble also has a cost of 1: the previous position of the red pebble has to be turned to blue (at a cost of 1), and then we can place a red pebble at the new position (since source nodes of the DAG can be recomputed at no cost).

However, the exact number of red pebbles that we have to move in this second step depends on whether $a$ and $b$ are neighbors in $G$. If $a$ and $b$ are adjacent, then the two input groups intersect in a node $v_{a,b}=v_{b,a}$, so we only have to move $N-1$ red pebbles; otherwise the two input groups are disjoint, and thus we have to move $N$ red pebbles. Thus altogether, the required cost of the operations between two consecutive group-visits is either $N+1$ or $N$, depending on whether the two nodes are adjacent in $G$. Note that a suboptimal pebbling might also execute further operations, but these steps are always necessary, so this provides a lower bound on the cost between two target computations.

Hence, the pebbling of the DAG corresponds to visiting the nodes of $G$ in some permutation $\pi$, with the cost of the pebbling directly dependent on the number of consecutive node pairs in $\pi$ that are connected in $G$. There exists a pebbling strategy of cost at most $(N-1) \cdot N$ only if there is a permutation $\pi$ such that each pair of consecutive nodes is connected, i.e. if there exists a Hamiltonian Path in $G$. However, if a Hamiltonian Path exists, then visiting the input groups in this order indeed gives a pebbling of cost $(N-1) \cdot N$, which completes our reduction. \qedhere
\end{proof}

Slightly modified versions of this construction provide the same complexity result in the remaining models. The \textsc{oneshot} model behaves the same way, except that we are allowed to delete the red pebbles on the source nodes after their last use, and thus we have to reduce the allowed cost accordingly. In the \textsc{base} model, we further insert a H2C gadget, and account for the extra cost it introduces; this makes it equivalent to the \textsc{oneshot} model. Finally, the DAG of the \textsc{base} model also suffices for the \textsc{compcost} model, provided that we further increase the allowed cost by $\epsilon$ for each compute operation.

We again point out that the problem in the \textsc{nodel} models was already proven to be NP-complete before \cite{RBcomplex}.

\section{$\delta$-inapproximability for $\delta<2$} \label{sec:VC}

We now show that in the \textsc{oneshot} model, approximating the optimal pebbling to a constant factor smaller than 2 also implies that the Vertex Cover problem is also approximable to a constant factor smaller than 2. However, if the unique games conjecture holds, then such an approximation is not possible for Vertex Cover \cite{VertexCoverHard}.

\begin{theorem} \label{th:vc}
For any $\delta<2$, there is no $\delta$-approximation algorithm for the \textsc{Pebbling} problem in the \textsc{oneshot} model, unless the unique games conjecture is falsified.
\end{theorem}

\renewcommand{\proofname}{Proof (\textit{with details in Appendix \ref{App:A:VC}}).}

\begin{proof} Assume we are given a graph $G$ on $N$ nodes. We again take every node $a$ of $G$, and we now create two distinct input groups for $a$; we refer to these as the first-level group of $a$ (denoted by $V_{a,1}$) and the second-level group of $a$ (denoted by $V_{a,2}$). Both of these groups have the same size $k$ for every $a$; we choose the parameter $k$ such that $k=\omega(N^2)$.

\begin{figure*}
	\centering
	\input{vertexcover.tikz}
	\caption{Illustration of the input groups created for two adjacent nodes $a$ and $b$ of $G$ in Theorem \ref{th:vc}. For a simplified notation, we draw an arrow from an input group to a target node to show that all nodes in this input groups have an outgoing edge to the target node in our DAG.}
	\label{fig:VC}
\end{figure*}

\begin{figure*}
	\centering
	\input{vc_transform.tikz}
	\caption{DAG construction corresponding to an example graph $G$ in Theorem \ref{th:vc}. For a simplified notation, we only illustrate the input groups of the DAG, denoted by squares in the figure. We use arrows to show dependencies (i.e. a group $V_{a,1}$ has to be visited before group $V_{b,2}$ because $V_{b,2}$ contains a target node of $V_{a,1}$) between input groups, and we use double gray lines to show that two groups are almost identical (i.e. they have a large intersection of common nodes). Since the only relevance of target nodes is that they force the algorithm to visit an input group, they are not shown explicitly in this simplified figure. The dotted rectangle corresponds to the DAG shown in Figure \ref{fig:VC}.}
	\label{fig:VC_transform}
\end{figure*}

As before, the input groups on the second level will only have one target node $t_{a,2}$. However, each input group on the first level will have not only one, but $N-1$ distinct target nodes. Each of these target nodes will correspond to some other node $b$ of $G$ (i.e. $b \neq a$); we denote the target node of group $V_{a,1}$ corresponding to $b$ by $t_{a,1,b}$.

For each edge $(a,b)$ of $G$, the second-level group of $a$ in the construction includes the corresponding target node of the first-level group of $b$, i.e. $t_{a,1,b} \in V_{b,2}$. This ensures that any pebbling algorithm has to visit $V_{a,1}$ before visiting $V_{b,2}$ to compute this target node. On the other hand, if $a$ and $b$ are not neighbors in $G$, then $t_{a,1,b}$ is just a sink node that is not included in any input group.

Furthermore, for each node $a$, we ensure that most nodes of the first and second-level input groups of $a$ coincide. That is, we create $k-N$ so-called \textit{common nodes} for each $a$, and we include these nodes in both $V_{a,1}$ and $V_{a,2}$. Recall that $k=\omega(N^2)$, so these common nodes dominate the input groups asymptotically. Besides the common nodes, we have already inserted up to $N-1$ target nodes in the second-level input groups (depending on the degree of $a$), so at this point, all input groups have a cardinality between $k-N$ and $k-1$. We simply fill up each input group (on both levels) with distinct extra nodes to reach a cardinality of $k$.

The construction is illustrated on Figures \ref{fig:VC} and \ref{fig:VC_transform}. Since each input group of the construction is of size $k$, we study this DAG with a choice of $R=k+1$.

The base idea of the construction is as follows. In this DAG, every (first- or second-level) input group has to be visited at least once. However, the second-level group of $a$ can only be visited after the first-level groups of all neighbors $b$ of $a$ have been visited, since a target of $V_{b,1}$ is included in the second-level group $V_{a,2}$. With $k$ much larger than $N$, the groups $V_{a,1}$ and $V_{a,2}$ are almost identical (consisting mostly of the common nodes), and thus a pebbling algorithm can spare a lot of cost by visiting $V_{a,1}$ and $V_{a,2}$ consecutively. In fact, the remaining nodes are asymptotically irrelevant, so the total cost will be determined by the number of nodes $a$ for which we can visit $V_{a,1}$ and $V_{a,2}$ consecutively. 

Whenever we visit $V_{a,1}$ and $V_{a,2}$ consecutively, we can simply compute the common nodes of $a$ first (they are source nodes, so this happens for free), visit both groups, and then delete the red pebbles from these common nodes, which is also free.

On the other hand, if the visits of $V_{a,1}$ and $V_{a,2}$ are not consecutive, then we first have to compute their common nodes when visiting $V_{a,1}$, and then move the red pebbles away from these nodes. As we still need to visit these common nodes for $V_{a,2}$, and recomputation is not possible in the \textsc{oneshot} model, we have to turn all these common nodes blue, and then turn them back to red later when visiting $V_{a,2}$ (after which the red pebbles can again be deleted for free). Thus if the groups $V_{a,1}$ and $V_{a,2}$ are visited non-consecutively, then each of the $k-N$ common nodes in the groups incurs a cost of at least 2.

Hence pebbling the DAG comes down to the task of finding the largest set of nodes in $G$ for which the two groups can be visited consecutively. Recall that for a consecutive visit of $V_{a,1}$ and $V_{a,2}$, the first-level group of each neighbor of $a$ has to have been visited earlier. The optimal strategy thus is to visit the first-level groups of a small vertex cover $VC$ in $G$, then visit both groups of each node in the remaining large independent set $IS$ consecutively, and in the end visit the second-level groups of nodes in $VC$. The common nodes in $IS$ will then incur no cost, and thus the cost will be in the magnitude of $2 k \cdot |VC|$, due to the transfer operations on the common nodes of the input groups that are in $VC$. This cost is proportional to $|VC|$ (apart from a negligible $O(N^2)$ factor), so the cost of a pebbling directly corresponds to the size of the vertex cover defined by the non-consecutively visited groups. This shows that approximating the optimal pebbling cost to a $\delta$ factor in this construction also provides a $\delta$-approximation to the Vertex Cover problem in $G$. \qedhere
\end{proof}

There is no straightforward way to apply this proof to the rest of our models. In \textsc{nodel}, red pebbles from common nodes cannot be deleted, so common nodes incur a cost of $k$ even in the consecutive case; thus cost is proportional to $k \cdot (2 |VC| + N)$. In the \textsc{base} model, common nodes are either recomputable for free again (thus there is essentially no cost at all), or they are costly to compute in the first place, which also adds a term of $k \cdot N$ to the cost. The \textsc{compcost} model has the same problem as \textsc{base}, with a further cost of $\epsilon \cdot N^2$ added due to the computations.

\section{Inefficiency of greedy algorithms} \label{sec:greedy}

In the \textsc{oneshot} model, each node of the DAG can be computed only once. Hence, any pebbling strategy is essentially described by the (topological) order in which we decide to compute the nodes of the DAG (and besides this, our method to decide which red pebbles to take away from other nodes for this computation).

There are some straightforward greedy heuristics to define this ordering, i.e. to select the next node to compute in each step; it is natural to ask if such algorithms provide close-to-optimal solutions. Such greedy approaches could, for example, always select the node:
\begin{itemize}
	\setlength\itemsep{0.4ex}
	\item with the largest number of red pebbles among its inputs,
	\item with the smallest number of blue pebbles among its inputs,
	\item with the largest red pebbles to inputs ratio.
\end{itemize}

In all these approaches, the greedy choice always happens from the set of (yet uncomputed) nodes whose inputs have already all been computed, since these are the only candidates for the next node to compute. Note that these greedy methods only choose the next node to compute, but do not specify which red pebbles to move to its inputs; our examples will show that these algorithms are inefficient regardless of how the red pebbles are chosen.

This section shows that such greedy approaches can yield much higher cost than $\texttt{opt}(R)$. As previously, our constructions will consist of input groups of the same size $k$, and hence, each (non-source) node has the same indegree. For such graphs, the previous greedy approaches are all identical, so one example is enough to disprove the efficiency of all.

\renewcommand{\proofname}{Proof.}
\begin{theorem} \label{th:greedy}
In the \textsc{oneshot} model, there is a class of graphs and a choice of $R$ such that $\textsc{cost}_{\text{greedy}}(R) \geq \widetilde{\Theta}(\sqrt{n}) \cdot \textsc{opt}(R)$.
\end{theorem}

\begin{proof}

Our construction is in many ways similar to the one in the previous section: we create many pairs (in fact, even chains) of input groups with a large number of common nodes, which should therefore be visited consecutively in any optimal solution. However, we also create dependencies between these groups (i.e., include the target nodes of a group in another group) so that some have to be visited earlier than others. If a greedy approach follows another ordering and does not compute these target nodes in time, then it is unable to follow these consecutive chains and thus returns a solution with much higher cost.

Let us introduce a parameter $\ell$, the value of which will be chosen later. Our construction consists of $\binom{\ell+1}{2}$ input groups, aligned along a grid in the positions $(i, j)$ satisfying $1 \leq i, j \leq \ell$ and $i+j \leq \ell+1$, as shown in Figure \ref{fig:greedy}. As in our previous examples, each of these input groups have the same size, denoted by $k$. Along each diagonal (i.e. the groups with $i+j=x$ for some specific $x$), the input groups are essentially the same, i.e. they all consist of the same $k'$ common nodes, with $k' \! \approx \! k$. We create exactly one target node for each of these input groups, denoted by $t_{i,j}$ for the group at position $(i, j)$.

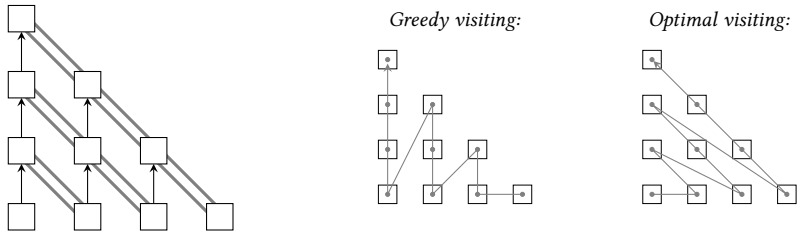
\begin{figure*}
	\centering
	\input{greedy.tikz}
	\caption{Construction for Theorem \ref{th:greedy}: grid of input groups for $\ell=4$, and their greedy and optimum visitation order. We apply the same notation as in Figure \ref{fig:VC_transform}.}
	\label{fig:greedy}
\end{figure*}

To generate dependencies between the groups, we also include the target node $t_{i,j}$ of each group $(i,j)$ in the group $(i, j+1)$, i.e. the group immediately above if such a group exists. This ensures that any algorithm has to visit group $(i,j)$ before visiting group $(i, j+1)$, since the target node of group $(i,j)$ is required in order to compute the target node of group $(i, j+1)$.

Furthermore, we add a few more extra nodes in order to misguide the greedy heuristic. For each $j$, we create a small intersection between the uppermost group of column $j$ and the lowermost group of column $j\!-\!1$ (i.e. groups $(j, \ell+1-j)$ and $(j-1, 1)$ for $j \in \{2, ..., \ell \}$).

Note that both for dependencies and misguidance, we only add a small constant number of nodes to each input group. Our choice of $k'$ will ensure that $k'$ is in a much larger magnitude, so each input group is still asymptotically dominated by the common nodes in the group. Also, the addition of these extra nodes leads to input groups of slightly different size; to fix this, we simply add an appropriate number of distinct extra nodes to each group in order to ensure that they all have the same cardinality $k=k'+O(1)$.

Finally, we need a technical step to ensure that the greedy approach enters this grid by visiting group $(\ell, 1)$ first. We add an extra input group $S_0$ below the whole grid, and create multiple target nodes for this group. We place a distinct target node of $S_0$ into each of the bottom groups $(i,1)$ of the grid. This way, any valid pebbling can only begin by visiting $S_0$, since the bottom groups have a dependence on $S_0$, and all other groups in the grid have a (direct or indirect) dependence on a bottom node. We then create a small intersection (of constantly many nodes) between $S_0$ and group $(\ell, 1)$ of the grid. This ensures that after visiting $S_0$, the greedy approach continues at group $(\ell, 1)$, since some of the inputs of group $(\ell, 1)$ already have a red pebble due to this intersection. Hence after computing all the target nodes of $S_0$, the greedy approach continues by visiting group $(\ell,1)$ of the grid.

Let us now analyze the progress of the greedy method in the grid. After $(\ell, 1)$, a greedy algorithm can only visit another bottom node next, since all other groups have an (indirect) dependency on the bottom group in the respective column. That is, even though group $(\ell, 1)$ has a large intersection with group $(\ell-1, 2)$, this group $(\ell-1, 2)$ also contains a target node of group $(\ell-1, 1)$, which is not computed yet. Thus the greedy algorithm has to choose among the bottom groups $(i,1)$; due to the small intersection we added, the algorithm chooses $(\ell-1, 1)$ next. At this point, the only available groups for visitation are the further bottom groups $(i,1)$ and the group $(\ell-1, 2)$ above the current position. Since no bottom group has a red pebble on any of its nodes, but group $(\ell-1, 2)$ contains the target node $t_{\ell-1, 1}$ which currently has a red pebble, the algorithm goes above to $(\ell-1, 2)$ next. Then the red pebbles are once again in an input group at the top of a column, so as before, the most attractive next move is to go to the next bottom node $(\ell-2, 1)$.

Following the small intersections sets in our design, the greedy method will continue in a similar fashion, visiting the columns from right to left, and processing each column in bottom-to-top direction. Whenever the algorithm is at the top of a column, the misguidance nodes ensure that it selects the bottom group of the next column in the following step, and whenever it is within a column, the only already enabled input group that has a red pebble on one of its nodes is the group immediately above. At every group $(i, j)$ in the process, moving up along the diagonal to $(i-1,j+1)$ would be a much more attractive move, but the algorithm cannot execute this because group $(i-1, j)$ is still unvisited at this point, so one of the nodes in group $(i-1,j+1)$ is not yet computed. Thus, without a deeper understanding of the dependencies between the groups, the greedy method follows a suboptimal ordering.

In contrast to this, after visiting $S_0$, an optimum pebbling could instead begin with group $(1, 1)$ first, then $(2, 1)$ and $(1,2)$, followed by $(3, 1)$, $(2,2)$ and $(1,3)$, and so on, always selecting the next bottom node $(i, 1)$ with $i$ increasing, and then processing the diagonal from this node up to $(1,i)$. This visitation order ensures that all groups of the diagonal in question are enabled by the time the algorithm reaches group $(i, 1)$.

If $k$ is large and $k'$ is only slightly smaller, then the cost of the pebbling is determined by the cost of revisiting the $k'$ common nodes in the groups. We have $\ell$ large sets of common nodes, corresponding to the $\ell$ diagonals. The greedy algorithm visits these common nodes again and again: the nodes in diagonal $i + j = \ell+1$ are visited $\ell$ times (once for each group in the diagonal), the nodes in diagonal $i + j = \ell$ are visited $\ell - 1$ times, and so on. Since all red pebbles are needed for every visitation, the greedy algorithm must turn these common nodes red and then blue each time when visiting the diagonal. As a technical detail, note that the first and last visitations are somewhat of an exception to this: the first time when the common nodes are visited, they obtain a red pebble free of charge since they are computed, and at the last visitation, the red pebbles from the nodes can be deleted free of charge. Due to this, the $\ell$ visitations of the common nodes in diagonal $i + j = \ell+1$ for example only require these transfer operations $\ell-1$ times instead of $\ell$; however, this is asymptotically irrelevant. Altogether, these transfer operations mean that the greedy algorithm will incur a cost of $2 k' \cdot (\ell\!-\!1 + \ell\! -\!2 + ...) = 2 k' \cdot \Theta(\ell^2)$.

On the other hand, an optimal solution simply computes the common nodes when visiting group $(i, 1)$, and deletes them after leaving the diagonal, doing both for free. Hence, the common nodes incur no cost in this optimal pebbling, and the cost of the pebbling is only determined by the operations on the remaining few nodes in each input group. Since there are only $k - k'$ such nodes in each of the $\binom{\ell+1}{2}$ input groups, and we only execute constantly many operations on each of them, the cost of the optimal pebbling is only in the magnitude of $(k - k') \cdot \Theta(\ell^2)$.

Let us now choose our parameters: we choose $\ell=\omega(1)$ as a slowly growing function (e.g. $\ell=\log{n}$), let $k=\widetilde{\Theta}(n)$, and let $k'=k-O(1)$ for a sufficiently large constant to allow the previously discussed extra nodes in each group. Since we have $O(k' \cdot \ell)$ common nodes and $O((k-k') \cdot \ell^2)$ extra nodes in the construction, the overall number of nodes is indeed less than $n$ if we choose $k=\Theta\left(\frac{n}{\ell}\right)$ with the appropriate constant.

This means that the optimum algorithm has a cost of $O(\log^2{n})$ in our construction, while the greedy algorithm incurs a cost of $\widetilde{\Theta}(n)$. Altogether, this provides a factor of $\widetilde{\Theta}(n)$ difference between the greedy solution and the optimum.

However, we note that our transformation to reduce $\Delta$ to a constant requires a slightly different choice of parameters in our construction; this reduces the difference to a factor of $\widetilde{\Theta}(\sqrt{n})$ in the constant-degree case. \qedhere
\end{proof}

Since the remaining models allow recomputation, defining the same greedy algorithm is already not a straightforward task in these models: an algorithm might decide to compute a node multiple times during a pebbling. We discuss the interpretation of the greedy rule in these remaining models in more detail in Appendix \ref{App:A:Gr}. With a slightly refined definition of the greedy algorithm, we can also adapt our construction to the remaining models.

In the \textsc{base} model, we need to add a H2C gadget to the construction to ensure that common nodes cannot be recomputed for free. However, this means that the first computation of the common nodes will also incur a considerable cost in the optimum case, which reduces the difference between the greedy and the optimum pebbling to a $\Theta(n^{1/3})$ factor, or a $\Theta(n^{1/6})$ factor when restricted to $\Delta=O(1)$.
In \textsc{nodel} and \textsc{compcost}, recall that the cost of two pebblings can only differ by a constant factor; in these cases, we show that this difference can become an arbitrarily high constant. The details of these modifications are discussed in Appendix \ref{App:A:Gr}.

\bibliographystyle{ACM-Reference-Format}
\bibliography{references}

\appendix

\section{More details on the constructions} \label{App:A}

\subsection{The tradeoff diagram in other models} \label{App:A:tradeoff}

It is rather straightforward to see that in the construction of Figure \ref{fig:tradeoff_graph}, the tradeoff function is also essentially the same for other models. In the \textsc{nodel} model, the only difference is the fact that chain nodes are never deleted, but turned to blue, which simply adds an offset of $n$ to the function. Similarly, the computation of the chain nodes adds an offset of $\epsilon \cdot n$ in the \textsc{compcost} case.

The only slightly more involved step is the addition of the H2C gadget, which is required in the \textsc{base} and \textsc{compcost} models. Note that the H2C gadget in general was designed to work for a specific $R$ value; however, as opposed to the other constructions in the paper, the tradeoff graph is a construction that we want to study for multiple $R$ values. Nonetheless, an easy modification allows us to adapt the H2C gadget for this case. Let the gadget consist of $d+1$ nodes in group $B$ as before, but let us now add $d+3$ starter nodes for each source of the DAG (instead of adding only 3, as we do in every other application of the gadget). This ensures that the H2C gadget can still be pebbled even with the minimal $R=d+2$, but since it has many target nodes, we still have to save at least $2$ of these into slow memory, even if $R=2d+2$ (i.e. the maximal $R$ we consider). For any $R<2d+2$, this will only mean that computing the control group nodes is slightly more costly, but this still only adds an overhead of $2d \cdot O(1)$ to the total cost.

Note that in models allowing recomputation, it is not a straightforward fact anymore that the cost can only decrease by $2n$ with the addition of each new red pebble. Hence in other models, it might also be possible to generate tradeoff functions that decrease more quickly than the maximal rate in the \textsc{oneshot} model.

\subsection{NP-hardness (Hamiltonian Path reduction) in other models} \label{App:A:NP}

The reduction in the \textsc{nodel} model has already been described in Section \ref{sec:NP}. We now describe the modifications required for the remaining models.

Compared to \textsc{nodel}, the \textsc{oneshot} model only requires us to modify the maximal allowed cost. This consists of two different steps. Firstly, note that in the \textsc{oneshot} model, moving from one input group to another does not have a cost of $N$ or $N-1$ anymore, but a cost of $2N$ or $2(N-1)$: since overriding the source nodes through recomputation is not possible, we also have to apply a transfer operation to make each node in the next input group red, so moving each red pebble has a cost of 2 instead of 1. Thus, the allowed cost of the pebbling should be modified to $(N-1) \cdot (2N-1)$.

Besides this, note that in \textsc{nodel}, we had to turn almost all the $N \cdot (N-1) - M$ source nodes blue by the time of computing the final target node (except for the input group visited last, where we can leave $N-1$ red pebbles when finishing the pebbling). In the \textsc{oneshot} model, the red pebble can be removed from each of these source nodes after they have been used for the last time. This reduces the cost by $N \cdot (N-1) - M - (N-1) = (N-1)^2 - M$. The target nodes, on the other hand, still have to receive a blue pebble while the remaining targets are being computed, so managing them has the same cost as in the \textsc{nodel} model. Altogether, this results in an allowed cost of $(N-1) \cdot (2N-1) - (N-1)^2 + M = (N-1) \cdot N + M$.

In the \textsc{base} and \textsc{compcost} models, we have to add a H2C gadget to disable the recomputation of source nodes. For simpler analysis, we now use the H2C gadget slightly differently than in other cases: besides the starter nodes, we also instantiate the group $B$ separately for each source node of the construction. With this, the computation of each source node becomes a separate process, each requiring all red pebbles and having a cost of exactly 4 in the base model (and a cost of exactly $4+(R+4) \cdot \epsilon$ in \textsc{compcost}), independently of other source nodes. Hence, we can compute each source node exactly once, and regardless of when it is done, the allowed cost in the reduction has to be increased by $4 \cdot (N \cdot (N-1) - M)$ in the \textsc{base} model, and $(4+(R+4) \cdot \epsilon) \cdot (N \cdot (N-1) - M)$ in \textsc{compcost}.

After the H2C gadgets are added, the reductions work the same way in these models as in the \textsc{oneshot} model, since we have the same DAG with the free recomputation of source nodes disabled. In \textsc{compcost}, the computation of each target nodes requires an extra cost of $\epsilon$ to further add to the cost parameters. Hence, with the appropriately defined cost limit, the reduction from Hamiltonian Path works in each of the models.

\subsection{Details of the reduction from Vertex Cover} \label{App:A:VC}

In the reduction from Vertex Cover, let us first discuss the cost of a pebbling in more detail.

Target nodes of input groups in the first level are first computed (when visiting the input group), then possibly turned to blue once and later back to red (if the group is not followed by the second-level input group that includes the target node), and deleted after the second-level input group is visited. Thus, such nodes incur a cost of at most 2; with $N \cdot (N-1)$ such target nodes altogether, this sums up to a cost of at most $2 \cdot N \cdot (N-1)$.

Target nodes of input groups in the second level are only computed once, and possibly turned to blue (in order to free the red pebble on them). With $N$ such target nodes, this amounts to a cost of at most $N$.

Filler nodes (inserted to increase the cardinality of each input group to $k$) are only computed once, and then can be deleted, so they incur no cost.

Finally the common nodes, as discussed earlier, are computed when the first-level group is visited, possibly turned to blue and then back to red (if the visiting of the two groups does not happen consecutively), and then deleted for free. Thus depending on the number of consecutive visits, the common nodes incur a cost of $|VC| \cdot k'$, where $k'$ denotes the number of common nodes in an input group (i.e. $k'=k-N$).

Note that these describe the maximal number of transfer moves that any reasonable pebbling would execute on the nodes in question. If an approximation algorithm executes some additional transfer moves on some nodes, then it can be further simplified to obtain an approximation algorithm with even lower cost. Thus, we can assume that our approximation algorithm only executes these moves. As such, the algorithm has a total cost of $|VC| \cdot k' + O(N^2)$. Thus by choosing, e.g., $k'=\Theta(N^3)$, the cost of the pebbling is indeed determined by the size of the vertex cover found.

Formally, the proof of inapproximability is as follows. In any valid pebbling, each input group is visited at least once, since all target nodes have to be computed at some point, and all red pebbles are required to compute the target nodes of an input group. Given a pebbling, consider the sequence obtained by only taking the first visitation of each input group. If $a$ and $b$ are neighboring nodes of $G$, then the first-level group of $a$ precedes the second-level group of $b$ in this sequence. Furthermore, we can assume that the first-level group of node $a$ always precedes the second-level group of $a$; otherwise, we can simply swap them without changing the magnitude of the total cost.

In this sequence, each time when the first-level group of $a$ is not followed by the second-level group of $a$ immediately, a cost of at least $k'$ is incurred, since the red pebbles from the first-level group cannot be deleted yet, but they are needed elsewhere for the common nodes of the next group in the sequence, so they all have to be turned blue. Hence, the minimal possible number of non-consecutively appearing group-pairs (multiplied by $2k'$) is a lower bound on the cost of any pebbling. If $a$ and $b$ are neighbors in $G$, then at most one of them can have its two groups consecutively; thus, the set of nodes for which the groups are consecutive forms an independent set. If we find a large independent set in $G$, then the complement of this set is a small vertex cover $VC$ of $G$. Thus if $VC_0$ is the smallest vertex cover, then $2k' \cdot |VC_0|$ is a lower bound on the cost of any pebbling strategy in our DAG. On the other hand, as discussed, a pebbling of cost $2k' \cdot |VC_0| + O(N^2)$ is indeed always possible.

Thus assume we have a $\delta$-approximation algorithm (with $\delta < 2$) for the pebbling problem, and we run it on our construction. Since running time is dominated by the $|VC| \cdot k'$ term discussed above, the remaining terms in $O(N^2)$ become irrelevant asymptotically; that is, for $N$ large enough, a $\delta$-approximation for every DAG is only possible if the approximation algorithm is always able to find a pebbling where the first- and second-level groups only appear non-consecutively for at most $\delta \cdot |VC_0|$ nodes of $G$. Since the remaining nodes form an independent set, this implies that the set of non-consecutive input groups provides a vertex cover of size at most $\delta \cdot |VC_0|$ in $G$. Since the DAG construction is polynomial in $N$, this implies that a polynomial time $\delta$-approximation algorithm for the pebbling problem would also provide a polynomial time $\delta$-approximation for Vertex Cover, contradicting the unique games conjecture. Thus if the conjecture holds, then such an algorithm cannot exist.

Finally, note that our reduction requires a relatively large number of red pebbles: as $k'=\omega(N^2)$, $k=k'+N$ and $R=k+1$, the pebbling problem obtained has an input parameter $R=\omega(N^2)$. However, this $R$ is only large in terms of $N$, the size of the original graph in the Vertex Cover problem. On the other hand, $R$ can be significantly smaller than the size of our DAG: we are free to add a large number of further nodes to the DAG that incur no cost and have no effect on the pebbling of the described construction. As the DAG size only needs to remain polynomial in $N$, this shows that for any integer $c$, we can have a reduction to a pebbling problem with $R=O(n^{\frac{1}{c}})$, with $n$ being the size of the DAG. Thus, our reduction also works if $R$ in the pebbling problem is restricted to smaller values.

\subsection{The greedy method in other models} \label{App:A:Gr}

In fact, defining a greedy algorithm in the remaining models (besides \textsc{oneshot}) is already not straightforward. With recomputations allowed, even if we choose the next node to compute, we need to decide whether to obtain the inputs from blue pebbles or through recomputation (in fact, even earlier: when we move a red pebble away from a node, we have to decide whether to turn it to blue or delete it).

However, we can adapt the construction to these models if we interpret the greedy approach as an ordering of the very first computation of nodes. Similarly to before, we can add a H2C gadget to ensure that recomputations are not beneficial in any case. In this interpretation, our results hold for any pebbling that uses a greedy heuristic to process the DAG, even if it is a `clever greedy' algorithm in the sense that for each new computation, it knows (from an oracle) the cheapest method to compute the specific node. In this sense, our results also hold in the other models, even for such moderately intelligent greedy heuristics.

Let us now consider the \textsc{nodel} and \textsc{compcost} models. In these models, the pebbling of the graph has an inherent cost of $\Omega(n)$, and hence we need a different choice of parameters. Let $k$ be a large constant value, $k'$ be a slightly smaller constant, and $\ell=\Theta(\sqrt{n})$. Recall that the number of nodes in our greedy construction is in the magnitude of $\ell^2 \cdot k$.

This implies that the common nodes of groups now only amount to $k' \cdot \ell = O(\sqrt{n})$ nodes in the graphs, so most nodes of the graph are the extra nodes in the $\binom{\ell+1}{2}$ groups. These can either be target nodes of the input group below, or nodes inserted into uppermost/lowermost groups of a column to misguide the greedy heuristic, or further filler nodes to ensure that each input group has the same size. Note that the number of misguiding nodes is only $O(1)$ in each column, and by merging the filler nodes in the different groups of each column, we can also ensure that we only use $O(1)$ filler nodes in each column. This reduces both the number of filler nodes and misguiding nodes to $O(1) \cdot \ell = O(\sqrt{n})$, so most nodes of the graph are indeed target nodes.

This already limits the number of source nodes in the graph to $O(\sqrt{n})$, allowing us to add a H2C gadget to the construction without any significant effect on the total number of nodes. Since almost all nodes of the graph ($n-O(\sqrt{n})$ of them) are now non-source nodes, the minimal cost of a pebbling in the \textsc{compcost} model is indeed in the magnitude of $\epsilon \cdot n$. Also, since the number of red pebbles $k+1$ is a constant, the minimal pebbling cost is $n-O(1) \approx n$ in the \textsc{nodel} model. This shows that enabling the source nodes through the H2C gadget is not a dominant part of the pebbling cost.

Recall that in our construction, the greedy algorithm has a cost in the magnitude of $k \cdot \ell^2$, while the optimum cost is in the magnitude of $(k-k') \cdot \ell^2$. Choosing a large constant $k$ and a $k'$ that is only slightly smaller,  we can create an arbitrary constant factor difference between the optimum and greedy pebblings. As the cost of any two pebblings are within a constant factor, this amounts to a significant difference in these models. In particular in the \textsc{nodel} model, this allows us to reach a factor $2 \Delta - c$ difference for a small constant $c$, essentially reaching the maximal difference factor of $2 \Delta + 1$ for larger $\Delta$ values.

Finally, let us consider the \textsc{base} model. We also add a H2C gadget here to ensure that source nodes (and most importantly among these, common nodes) are never recomputed. This adds an extra cost of $O(1) \cdot k \cdot \ell$, since these nodes initially have to be computed at some point. With this, the cost of greedy pebbling is now in the magnitude of $\ell \cdot k + \ell^2 \cdot k$, thus it is still dominated by the progress through the grid. The cost of the optimum is in the magnitude of $\ell \cdot k + \ell^2 \cdot (k-k')$. Hence by choosing $\ell = \Theta(\sqrt[3]{n})$, $k=\Theta(\sqrt[3]{n})$ and $k'=k-O(1)$, we get a cost of $\Theta(n)$ for the greedy pebbling and $\Theta(n^{\frac{2}{3}})$ for the optimum. Thus in the \textsc{base} model, there can be a factor of $\Theta(\sqrt[3]{n})$ difference between the two algorithms.

\section{Restriction to constant indegree} \label{App:B}

The main idea for the constant degree gadget (CD gadget) has already been outlined before. If $R-1$ red pebbles are placed on the leftmost nodes, then using two more pebbles, the entire gadget can be pebbled for free (in the \textsc{oneshot} and \textsc{base} models), regardless of the value of $h$. However, if we want to pebble the gadget using fewer than $R-1$ red pebbles on the leftmost nodes, then the red pebbles have to be moved around within the leftmost nodes, incurring a cost of at least 2 for every layer (assuming that the leftmost nodes cannot be recomputed for free), and hence a total cost of at least $2h$. For $h$ large enough, this already ensures that any reasonable pebbling has to place a red pebble on all the leftmost nodes at some point.

Note that this modification also requires us to raise the number of available red pebbles from $R$ to $R+1$. When adapting our constructions to constant indegree, we transform each input group into a CD gadget, and thus the construction will behave as before with $R+1$ available red pebbles.

Note that each input group is used to compute some specific target nodes in our constructions. After replacing the input groups by CD gadgets, we can simply add these target nodes to the end of the CD gadget, drawing an edge from the last node in the last layer of the CD gadget to all the target nodes. This ensures that in any pebbling with reasonable cost, the target nodes can only be enabled after $R-1$ red pebbles have been placed in the leftmost nodes of the CD gadget.

\subsection{NP-hardness}

The constant degree gadgets are rather straightforward to use in the Hamiltonian Path reduction. Recall that the maximal allowed cost $C$ in each of the models is within $O(N^2)$. Hence a choice of $h=C+1$ (or much higher) already ensures that if any pebbling strategy does not have all red pebbled at an input group at some point in time, then computing the constant degree gadget has a cost of at least $C+1$, and thus it is invalid. Therefore, any valid pebbling still has to visit all the input groups (i.e., have all red pebbles in the input group and its targets at some point), and a pebbling again comes down to the order of visiting the groups. Note that the number of nodes in the DAG still remains within $O(N^2) \cdot h = O(N^4)$.

Note that to have constant indegree for all nodes of the graph, we also have to transform the H2C gadget which is added in some of the models. This happens exactly the same way as with other input groups: we add 3 CD gadgets for every source node, with the leftmost nodes of the CD gadget being group $B$ of the H2C gadget, and the target node being one of the 3 starter nodes. Hence the computation of any starter node (in reasonable cost) also requires us to place all red pebbles on group $B$ of the H2C gadget.

Note that this is already sufficient for the \textsc{oneshot} and \textsc{base} models, where, after all red pebbles are moved to a CD gadget, the computation of the $h$ layers of the gadget is completely free. The other two models also require some modification in the allowed cost $C$. In \textsc{nodel}, we have to turn each node of each layer blue in every CD gadget, which raises the total cost by $(R-1) \cdot h$ for each added CD gadget. In \textsc{compcost}, while the intermediate nodes can be deleted for free again, the computation of the layers in each CD gadget amounts to an extra cost of $(R-1) \cdot h \cdot \epsilon$.

\subsection{Inapproximability}

In the reduction from Vertex Cover, we also simply have to select $h$ large enough in order to use CD gadgets. Since the optimum cost is always at most $O(N \cdot k)$, choosing $h=\Theta(N^2 \cdot k^2)$ ensures that any node not placing all red pebbles in an input group simultaneously (and thus incurring a cost of at least $h$) is not a constant-factor approximation.

Note that we only consider this construction in the \textsc{oneshot} model, where processing a CD gadget is free; therefore, replacing input groups by CD gadgets does not modify the cost of any (reasonable) pebbling strategy.

\subsection{Construction for the greedy algorithm}

The method, however, requires some modifications to be applicable to our results on the greedy algorithm.

First of all, note that in the greedy construction, the transformation is not even required for the case of the \textsc{nodel} and \textsc{compcost} models. In these models, we have selected $k=O(1)$, which already ensures that every target node in the construction has constant indegree.

Let us now consider the construction in the other two models. Note that previously, since each computation required the use of all red pebbles, it did not matter how the greedy algorithm selects the red pebbles to use for a specific computation. This becomes a more significant question in our construction for $\Delta=O(1)$. However, using CD gadgets with a high $h$ value essentially ensures that we always study the most clever greedy algorithm that select red pebbles in an optimal way. Such a clever greedy algorithm always removes the red pebbles from the already processed input groups into the input group that is currently being processed; this way, computing all nodes of the CD gadget happens at no cost. In any other case, the computation requires a cost of at least $h$, so if we select $h$ to be larger than the worst-case cost of the clever algorithm, then any suboptimal choice of red pebbles will certainly result in a higher total cost than the greedy algorithm that selects red pebbles optimally. Hence in this case, it is enough to compare the optimum pebbling to the greedy algorithm with the most intelligent choice of red pebbles.

For this construction, we slightly modify the original CD gadget. In the leftmost row of Figure \ref{fig:constdeg}, let us replace each node by a small group of constant size; the indegree of each node in the gadget will still remain a constant after this. This modification allows to take the $O(1)$ extra nodes (i.e., the ones besides the common nodes) in each input group, and add it to each small leftmost group in the gadget. Thus the only difference between the leftmost groups of the gadget will essentially be that each contains a separate common node of the specific input group.

The addition of a few more nodes to each leftmost group allows us to analyze the greedy algorithm more easily. Let us take different constants $c_1 > c_2 > c_3$. For each input group, let us create $c_1$ specific nodes that are included in each leftmost group of the CD gadget obtained from the input group. For each column of the grid, let us add $c_2$ nodes that are included in each input group in the given column of the grid. Finally, let us use $c_1$ nodes to misguide the greedy heuristic such that it prefers the bottom group of the next column after the top group of the previous column. This ensures that whenever the greedy algorithm starts processing an input group, it will prefer to continue with the next leftmost group of the gadget until all the leftmost groups are filled with red pebbles. After the entire CD gadget is pebbled, it will prefer to move to the input group above (in the grid), or the next column if it is already at the top of the current column. This ensures that the greedy algorithm follows the pattern discussed in Section \ref{sec:greedy} through the grid.

The extra H2C gadgets in the \textsc{base} model can be transformed as we have seen in the previous constructions.

Let us now discuss the choice of parameters required. In the \textsc{oneshot} model, we can choose $\ell$ to be a slowly growing function in $\omega(1)$ as before. In this case, the cost of the greedy algorithm can go up to as much as $k'$. Now if we select $k'=\widetilde{\Theta}(\sqrt{n})$, then a choice of $h=\widetilde{\Theta}(\sqrt{n})$ ensures that trying to pebble an input group without using all the red pebbles is already more costly than the total cost of the greedy method. Since the number of nodes in the graph is in the magnitude of $\ell^2 \cdot k \cdot h$, this is indeed a valid choice of parameters. As the optimum pebbling has a cost of $O((k-k') \cdot \ell^2)$, this shows that the greedy algorithm can be a factor $\widetilde{\Theta}(\sqrt{n})$ worse than the optimum.

In the \textsc{base} model, with $k'=k-O(1)$ as before, the cost of the greedy algorithm is essentially $\ell^2 \cdot k$, while the cost of the optimum pebbling is $\ell^2+\ell \cdot k$ (see Appendix \ref{App:A:Gr}). If we choose $h = \Theta(\ell^2 \cdot k)$, we can ensure that not using all red pebbles for an input group is always worse than the greedy cost of $\ell^2 \cdot k$. Since we have $\Theta(\ell^2 \cdot k \cdot h)$ nodes in the graph, we can choose $h=\Theta(\sqrt{n})$, $\ell = \Theta(\sqrt[6]{n})$ and $k = \Theta(\sqrt[6]{n})$. This results in a cost of $\Theta(\sqrt{n})$ for the greedy algorithm and a cost of $\Theta(\sqrt[3]{n})$ for the optimal pebbling, showing that the difference can go up to a factor of $\Theta(\sqrt[6]{n})$.

\section{Pebbling with different starting or finishing states} \label{App:C}

As mentioned in Section \ref{sec:basics}, some papers on the topic consider slightly different definitions for the initial and/or finishing state of pebblings. In particular, the original paper of Hong \textit{et al} \cite{RBintro} introducing the problem also uses these alternative definitions. We now show that for all our results, the settings are essentially equivalent, and applying the alternative definitions only result in an asymptotically insignificant difference in cost. 

In our definitions, we consider source nodes as regular nodes with 0 inputs, which are, thus, always computable for free. In contrast to this, some studies assume that source nodes already contain a blue pebble in the beginning, and are not computable at all; you have to turn these blue pebble to a red one in order to start computation. Note that we can also adapt our results to such a setting: as described in Section \ref{sec:basics} when analyzing the number of source nodes, we can add a single source $s_0$ to the DAG and make this the input of every other node. In this case, $s_0$ becomes the only source node of the graph, so the (possibly many) original sources can be computed for free as before. Turning $s_0$ to red in the beginning will induce an extra cost of 1, but this does not affect any of our results. This transformation into a single-source DAG already been described in \cite{RBcomplex} before.

Note that together with the observations of Section \ref{sec:basics}, this shows that both definitions of initialization (freely computable sources, or sources starting with blue pebbles) can model both possible situations of interest (when source nodes are free to compute, and when source nodes should incur some cost). If source nodes are free to compute, we can still make them incur some cost using a H2C gadget (as in Section \ref{sec:basics}); and if source nodes start with blue pebbles, we can still reduce the incurred cost to 1 using a single source $s_0$ (as shown above).

Also, while in our definition, we consider a pebbling finished when all sink nodes contain a pebble (either red or blue), some papers explicitly require each sink node to contain a blue pebble in the end of the pebbling. Note, however, that once each sink node has a pebble of some color, we can turn all these pebbles blue at an extra cost of at most 1 per sink node. One can observe that in each of our constructions, the number of sink nodes is always asymptotically smaller than the cost of the optimal pebbling, and thus this extra cost of 1 per sink node has no effect on the magnitude of the total cost. Hence, our results also hold if all sinks are required to have a blue pebble in the end.

More specifically, the construction of Section \ref{sec:tradeoff} only has 1 sink node, and an extra cost of 1 does not have any major effect on the behavior of the tradeoff. For the Hamiltonian Path reduction, any valid pebbling has to turn all but 1 sink node blue in any case; thus the different setting also only increases the cost by 1 (thus we also have to increment the maximal allowed cost by 1). In the Vertex Cover reduction, the only sink nodes are the target nodes of second-level groups, thus the DAG contains $N$ sink nodes. Since we have specifically chosen $k$ much larger than $N$, unless $G$ contains a vertex cover of size 0, an increase of $N$ in the cost has no effect on the reductions. Finally, the construction for the greedy algorithm has $\ell$ sink nodes, and a cost increase of $\ell$ has no effect asymptotically in any of the models.

\end{document}

%% file: constdeg.tikz
\begin{tikzpicture}
	
	\draw[line width=0.3pt, arrows=-stealth] (25pt,-27pt) -- (62pt,-3pt);
	\draw[line width=0.3pt, arrows=-stealth] (25pt,-14pt) -- (61pt,-1pt);
	\draw[line width=0.3pt, arrows=-stealth] (25pt,14pt) -- (61pt,1pt);
	\draw[line width=0.3pt, arrows=-stealth] (25pt,27pt) -- (62pt,3pt);
	
	\draw[black, fill=white] (25pt,-27pt) circle (0.9ex);
	\draw[black, fill=white] (25pt,-14pt) circle (0.9ex);
	\node[anchor=center] at (25pt,0pt) {\small .};
	\node[anchor=center] at (25pt,4pt) {\small .};
	\node[anchor=center] at (25pt,-4pt) {\small .};
	\draw[black, fill=white] (25pt,14pt) circle (0.9ex);
	\draw[black, fill=white] (25pt,27pt) circle (0.9ex);
	
	\draw[black, fill=white] (65pt,0pt) circle (0.9ex);
	
	\draw[gray, very thick, dashed] (19pt,-35pt) -- (19pt,35pt) -- (31pt,35pt) -- (31pt,-35pt) -- cycle;
	
	\node[anchor=north] at (25pt,-37pt) {\scriptsize group of};
	\node[anchor=north] at (25pt,-46pt) {\scriptsize $R\!-\!1$ nodes};
	
	\node[anchor=north] at (68.5pt,-1.5pt) {\footnotesize $t$};
	
	\draw[very thick, gray, arrows=-stealth] (80pt,0pt) -- (110pt,0pt);

	\node[anchor=south] at (95pt,5pt) {\scriptsize \textit{replaced}};
	\node[anchor=south] at (95pt,-1pt) {\scriptsize \textit{by}};

	\draw[line width=0.3pt, arrows=-stealth] (125pt,-27pt) -- (151pt,-27pt);
	\draw[line width=0.3pt, arrows=-stealth] (125pt,-14pt) -- (156pt,-14pt);
	\draw[line width=0.3pt, arrows=-stealth] (125pt,14pt) -- (166pt,14pt);
	\draw[line width=0.3pt, arrows=-stealth] (125pt,27pt) -- (171pt,27pt);
	
	\draw[line width=0.3pt, arrows=-stealth] (155pt,-27pt) -- (158.5pt,-17.5pt);
	\draw[line width=0.3pt, arrows=-stealth] (160pt,-14pt) -- (162.5pt,-7pt);
	\draw[line width=0.3pt, arrows=-stealth] (166.5pt,5.5pt) -- (168.5pt,11pt);
	\draw[line width=0.3pt, arrows=-stealth] (170pt,14pt) -- (173.5pt,23.5pt);
	\draw[line width=0.3pt, arrows=-stealth] (175pt,27pt) -- (183pt,-23pt);
	
	\draw[line width=0.2pt, arrows=-stealth] (125pt,-27pt) -- (130pt,-34pt) -- (160pt,-34pt) -- (165pt,-27pt) -- (181pt,-27pt);
	\draw[line width=0.2pt, arrows=-stealth] (125pt,-14pt) -- (130pt,-21pt) -- (165pt,-21pt) -- (170pt,-14pt) -- (186pt,-14pt);
	\draw[line width=0.2pt, arrows=-stealth] (125pt,14pt) -- (130pt,7pt) -- (175pt,7pt) -- (180pt,14pt) -- (196pt,14pt);
	\draw[line width=0.2pt, arrows=-stealth] (125pt,27pt) -- (130pt,20pt) -- (180pt,20pt) -- (185pt,27pt) -- (201pt,27pt);
	
	\draw[line width=0.3pt, arrows=-stealth] (185pt,-27pt) -- (188.5pt,-17.5pt);
	\draw[line width=0.3pt, arrows=-stealth] (190pt,-14pt) -- (192.5pt,-7pt);
	\draw[line width=0.3pt, arrows=-stealth] (196.5pt,5.5pt) -- (198.5pt,11pt);
	\draw[line width=0.3pt, arrows=-stealth] (200pt,14pt) -- (203.5pt,23.5pt);
	\draw[line width=0.3pt, arrows=-stealth] (205pt,27pt) -- (213pt,-23pt);
	
	\draw[black, fill=white] (125pt,-27pt) circle (0.9ex);
	\draw[black, fill=white] (125pt,-14pt) circle (0.9ex);
	\node[anchor=center] at (125pt,0pt) {\small .};
	\node[anchor=center] at (125pt,4pt) {\small .};
	\node[anchor=center] at (125pt,-4pt) {\small .};
	\draw[black, fill=white] (125pt,14pt) circle (0.9ex);
	\draw[black, fill=white] (125pt,27pt) circle (0.9ex);
	
	\draw[gray, very thick, dashed] (119pt,-35pt) -- (119pt,35pt) -- (131pt,35pt) -- (131pt,-35pt) -- cycle;
	
	\node[anchor=north] at (125pt,-37pt) {\scriptsize group of};
	\node[anchor=north] at (125pt,-46pt) {\scriptsize $R\!-\!1$ nodes};
	
	\draw[black, fill=white] (155pt,-27pt) circle (0.9ex);
	\draw[black, fill=white] (160pt,-14pt) circle (0.9ex);
	\node[anchor=center] at (163.75pt,-4pt) {\small .};
	\node[anchor=center] at (165pt,0pt) {\small .};
	\node[anchor=center] at (166.25pt,4pt) {\small .};
	\draw[black, fill=white] (170pt,14pt) circle (0.9ex);
	\draw[black, fill=white] (175pt,27pt) circle (0.9ex);
	
	\draw[black, fill=white] (185pt,-27pt) circle (0.9ex);
	\draw[black, fill=white] (190pt,-14pt) circle (0.9ex);
	\node[anchor=center] at (193.75pt,-4pt) {\small .};
	\node[anchor=center] at (195pt,0pt) {\small .};
	\node[anchor=center] at (196.25pt,4pt) {\small .};
	\draw[black, fill=white] (200pt,14pt) circle (0.9ex);
	\draw[black, fill=white] (205pt,27pt) circle (0.9ex);
	
	\draw[line width=0.3pt, arrows=-stealth] (215pt,-27pt) -- (218.5pt,-17.5pt);
	\draw[line width=0.2pt, arrows=-stealth] (125pt,-27pt) -- (130pt,-36pt) -- (195pt,-36pt) -- (200pt,-27pt) -- (211pt,-27pt);
	\draw[black, fill=white] (215pt,-27pt) circle (0.9ex);
	
	\node[anchor=north] at (225pt,2pt) {\scriptsize . . . };
	
	\draw[thick, gray] (152pt,31pt) -- (152pt,34pt) -- (180pt,34pt) -- (180pt,31pt);
	\draw[thick, gray] (166pt,34pt) -- (166pt,36pt);
	\node[anchor=south] at (166pt,34pt) {\scriptsize 1 layer};
	
	\draw[thick, gray] (150pt,-40pt) -- (150pt,-43pt) -- (240pt,-43pt) -- (240pt,-40pt);
	\draw[thick, gray] (195pt,-43pt) -- (195pt,-45pt);
	\node[anchor=north] at (195pt,-43pt) {\scriptsize $h$ layers};

\end{tikzpicture}

%% file: h2c.tikz
\begin{tikzpicture}

	\draw[arrows=-stealth] (0pt,0pt) -- (21pt,-25pt);
	\draw[arrows=-stealth] (0pt,0pt) -- (21pt,-13pt);
	\draw[arrows=-stealth] (0pt,0pt) -- (21pt,13pt);
	\draw[arrows=-stealth] (0pt,0pt) -- (21pt,25pt);
	
	\draw[line width=0.3pt, arrows=-stealth] (25pt,-27pt) -- (61pt,-22pt);
	\draw[line width=0.3pt, arrows=-stealth] (25pt,-14pt) -- (60.5pt,-20pt);
	\draw[line width=0.3pt, arrows=-stealth] (25pt,14pt) -- (61pt,-18.5pt);
	\draw[line width=0.3pt, arrows=-stealth] (25pt,27pt) -- (61.5pt,-17pt);
	\draw[line width=0.3pt, arrows=-stealth] (25pt,-27pt) -- (62pt,-3pt);
	\draw[line width=0.3pt, arrows=-stealth] (25pt,-14pt) -- (61pt,-1pt);
	\draw[line width=0.3pt, arrows=-stealth] (25pt,14pt) -- (61pt,1pt);
	\draw[line width=0.3pt, arrows=-stealth] (25pt,27pt) -- (62pt,3pt);
	\draw[line width=0.3pt, arrows=-stealth] (25pt,-27pt) -- (61.5pt,17pt);
	\draw[line width=0.3pt, arrows=-stealth] (25pt,-14pt) -- (61pt,18.5pt);
	\draw[line width=0.3pt, arrows=-stealth] (25pt,14pt) -- (60.5pt,20pt);
	\draw[line width=0.3pt, arrows=-stealth] (25pt,27pt) -- (61pt,22pt);
	
	\draw[arrows=-stealth] (65pt,-20pt) -- (96.5pt,-2pt);
	\draw[arrows=-stealth] (65pt,0pt) -- (96pt,0pt);
	\draw[arrows=-stealth] (65pt,20pt) -- (96.5pt,2pt);

	\draw[black, fill=white] (0pt,0pt) circle (0.9ex);
	
	\draw[black, fill=white] (25pt,-27pt) circle (0.9ex);
	\draw[black, fill=white] (25pt,-14pt) circle (0.9ex);
	\node[anchor=center] at (25pt,0pt) {\small .};
	\node[anchor=center] at (25pt,5pt) {\small .};
	\node[anchor=center] at (25pt,-5pt) {\small .};
	\draw[black, fill=white] (25pt,14pt) circle (0.9ex);
	\draw[black, fill=white] (25pt,27pt) circle (0.9ex);
	
	\draw[black, fill=white] (65pt,-20pt) circle (0.9ex);
	\draw[black, fill=white] (65pt,0pt) circle (0.9ex);
	\draw[black, fill=white] (65pt,20pt) circle (0.9ex);
	
	\draw[black, fill=white] (100pt,0pt) circle (0.9ex);
	
	\draw[gray, very thick, dashed] (19pt,-35pt) -- (19pt,35pt) -- (31pt,35pt) -- (31pt,-35pt) -- cycle;
	
	\node[anchor=north] at (-2.5pt,-2pt) {\scriptsize $s$};
	
	\node[anchor=north] at (25pt,-34pt) {\scriptsize group $B$ of};
	\node[anchor=north] at (25pt,-41pt) {\scriptsize $R\!-\!1$ nodes};
	
	\node[anchor=north] at (68.5pt,18.5pt) {\scriptsize $u_1$};
	\node[anchor=north] at (68.5pt,-1.5pt) {\scriptsize $u_2$};
	\node[anchor=north] at (68.5pt,-21.5pt) {\scriptsize $u_3$};
	
	\node[anchor=north] at (103pt,-1pt) {\scriptsize $v$};

\end{tikzpicture}

%% file: tradeoff.tikz
\begin{tikzpicture}
	
	\draw[line width=0.3pt] (0pt,10pt) -- (12pt,20pt);
	\draw[line width=0.3pt] (3pt,18pt) -- (12pt,20pt);
	\draw[line width=0.3pt] (3pt,22pt) -- (12pt,20pt);
	\draw[line width=0.3pt] (0pt,30pt) -- (12pt,20pt);
	
	\draw[line width=0.3pt] (0pt,-10pt) -- (12pt,-20pt);
	\draw[line width=0.3pt] (3pt,-18pt) -- (12pt,-20pt);
	\draw[line width=0.3pt] (3pt,-22pt) -- (12pt,-20pt);
	\draw[line width=0.3pt] (0pt,-30pt) -- (12pt,-20pt);
	
	\draw[line width=0.3pt] (12pt,20pt) -- (55pt,20pt);
	\draw[line width=0.3pt, arrows=-stealth] (15pt,20pt) -- (19pt,11pt);
	\draw[line width=0.3pt, arrows=-stealth] (35pt,20pt) -- (39pt,11pt);
	\draw[line width=0.3pt, arrows=-stealth] (55pt,20pt) -- (59pt,11pt);
	
	\draw[line width=0.3pt] (12pt,-20pt) -- (65pt,-20pt);
	\draw[line width=0.3pt, arrows=-stealth] (25pt,-20pt) -- (29pt,-11pt);
	\draw[line width=0.3pt, arrows=-stealth] (45pt,-20pt) -- (49pt,-11pt);
	\draw[line width=0.3pt, arrows=-stealth] (65pt,-20pt) -- (69pt,-11pt);
	
	\draw[line width=0.3pt, arrows=-stealth] (20pt,8pt) -- (28pt,-6pt);
	\draw[line width=0.3pt, arrows=-stealth] (30pt,-8pt) -- (38pt,6pt);
	\draw[line width=0.3pt, arrows=-stealth] (40pt,8pt) -- (48pt,-6pt);
	\draw[line width=0.3pt, arrows=-stealth] (50pt,-8pt) -- (58pt,6pt);
	\draw[line width=0.3pt, arrows=-stealth] (60pt,8pt) -- (68pt,-6pt);
	\draw[line width=0.3pt, arrows=-stealth] (70pt,-8pt) -- (75pt,0pt);
	\node[anchor=center] at (81pt,1pt) {\small ...};
	
	\draw[black, fill=white] (0pt,10pt) circle (0.7ex);
	\node[anchor=center] at (0pt,20pt) {\small .};
	\node[anchor=center] at (0pt,18pt) {\small .};
	\node[anchor=center] at (0pt,22pt) {\small .};
	\draw[black, fill=white] (0pt,30pt) circle (0.7ex);
	
	\draw[gray, very thick, dashed] (-5pt,5pt) -- (-5pt,35pt) -- (5pt,35pt) -- (5pt,5pt) -- cycle;
	
	\node[anchor=east] at (-4pt,23pt) {\scriptsize group of};
	\node[anchor=east] at (-4pt,17pt) {\scriptsize $d$ nodes};
	
	\draw[black, fill=white] (0pt,-10pt) circle (0.7ex);
	\node[anchor=center] at (0pt,-20pt) {\small .};
	\node[anchor=center] at (0pt,-18pt) {\small .};
	\node[anchor=center] at (0pt,-22pt) {\small .};
	\draw[black, fill=white] (0pt,-30pt) circle (0.7ex);
	
	\draw[gray, very thick, dashed] (-5pt,-5pt) -- (-5pt,-35pt) -- (5pt,-35pt) -- (5pt,-5pt) -- cycle;
	
	\node[anchor=east] at (-4pt,-17pt) {\scriptsize group of};
	\node[anchor=east] at (-4pt,-23pt) {\scriptsize $d$ nodes};
	
	\draw[black, fill=white] (20pt,8pt) circle (0.7ex);
	\draw[black, fill=white] (30pt,-8pt) circle (0.7ex);
	\draw[black, fill=white] (40pt,8pt) circle (0.7ex);
	\draw[black, fill=white] (50pt,-8pt) circle (0.7ex);
	\draw[black, fill=white] (60pt,8pt) circle (0.7ex);
	\draw[black, fill=white] (70pt,-8pt) circle (0.7ex);

\end{tikzpicture}

%% file: diagram.tikz
\begin{tikzpicture}

	\draw[arrows=-stealth] (0pt,-3pt) -- (0pt,70pt);
	\draw[arrows=-stealth] (-3pt,0pt) -- (100pt,0pt);
	\draw[line width=0.2pt] (80pt,3pt) -- (80pt,-3pt);
	\draw[line width=0.2pt] (-3pt,60pt) -- (3pt,60pt);
	
	\node[anchor=east] at (-2pt,0pt) {\scriptsize $0$};
	\node[anchor=east] at (-2pt,60pt) {\scriptsize $2 d \! \cdot \! n$};
	\node[anchor=south] at (0pt,68pt) {\footnotesize $\texttt{opt}(R)$};
	
	\node[anchor=north] at (0pt,-2pt) {\scriptsize $d\!+\!2$};
	\node[anchor=north] at (80pt,-2pt) {\scriptsize $2d\!+\!2$};
	\node[anchor=west] at (99pt,0pt) {\footnotesize $R$};
	
	\draw[thick, midgray] (19pt,60pt) -- (21pt,60pt) -- (21pt,54pt) -- (19pt,54pt);
	\draw[thick, midgray] (21pt,57pt) -- (23pt,57pt);
	\node[anchor=west] at (20pt,57pt) {\scriptsize $2n$};
	
	\draw[thick, midgray] (72pt,14pt) -- (72pt,16pt) -- (80pt,16pt) -- (80pt,14pt);
	\draw[thick, midgray] (76pt,16pt) -- (76pt,18pt);
	\node[anchor=south] at (76pt,15pt) {\scriptsize $1$};
	
	\draw[darkgray, thick] (0pt,60pt) -- (8pt,60pt) -- (8pt,54pt) -- (16pt,54pt) -- (16pt,48pt) -- (24pt,48pt) -- (24pt,42pt) -- (32pt,42pt) -- (32pt,36pt) -- (40pt,36pt) -- (40pt,30pt) -- (48pt,30pt) -- (48pt,24pt) -- (56pt,24pt) -- (56pt,18pt) -- (64pt,18pt) -- (64pt,12pt) -- (72pt,12pt) -- (72pt,6pt) -- (80pt,6pt) -- (80pt,0pt);	
	
\end{tikzpicture}

%% file: hampath.tikz
\begin{tikzpicture}
	
	\draw[darkgray, very thick] (-86pt,24pt) -- (-86pt,-24pt);
	\draw[very thick] (-80pt,7pt) -- (-92pt,-7pt);
	\draw[very thick] (-80pt,-7pt) -- (-92pt,7pt);
	\draw[black, fill=white] (-86pt,24pt) circle (0.8ex);
	\draw[black, fill=white] (-86pt,-24pt) circle (0.8ex);
	\node[anchor=north] at (-81pt,24pt) {\small $a$};
	\node[anchor=north] at (-81pt,-24pt) {\small $b$};
	
	\draw[very thick, gray, arrows=-stealth] (-65pt,0pt) -- (-25pt,0pt);
	\node[anchor=south] at (-45pt,5pt) {\scriptsize \textit{transformed}};
	\node[anchor=south] at (-45pt,-1pt) {\scriptsize \textit{to}};

	\draw[midgray, densely dotted, very thick] (-2pt,-10pt) -- (-8.5pt,-10pt) -- (-25pt,52pt);
	\node[anchor=south] at (-25pt,57pt) {\scriptsize $b$'s contact};
	\node[anchor=south] at (-25pt,50pt) {\scriptsize node for $a$};
	
	\draw[midgray, densely dotted, very thick] (2pt,10pt) -- (8.5pt,10pt) -- (25pt,52pt);
	\node[anchor=south] at (25pt,57pt) {\scriptsize $a$'s contact};
	\node[anchor=south] at (25pt,50pt) {\scriptsize node for $b$};
	
	\draw[line width=0.2pt, arrows=-stealth] (0pt,10pt) -- (38pt,21pt);
	\draw[line width=0.2pt, arrows=-stealth] (0pt,19pt) -- (37pt,22.5pt);
	\draw[line width=0.2pt, arrows=-stealth] (0pt,28pt) -- (37pt,24.5pt);
	\draw[line width=0.2pt, arrows=-stealth] (0pt,37pt) -- (38pt,26pt);
	
	\draw[line width=0.2pt, arrows=-stealth] (0pt,-10pt) -- (38pt,-21pt);
	\draw[line width=0.2pt, arrows=-stealth] (0pt,-19pt) -- (37pt,-22.5pt);
	\draw[line width=0.2pt, arrows=-stealth] (0pt,-28pt) -- (37pt,-24.5pt);
	\draw[line width=0.2pt, arrows=-stealth] (0pt,-37pt) -- (38pt,-26pt);
	
	\draw[black, fill=white] (0pt,10pt) circle (0.7ex);
	\draw[black, fill=white] (0pt,19pt) circle (0.7ex);
	\draw[black, fill=white] (0pt,28pt) circle (0.7ex);
	\draw[black, fill=white] (0pt,37pt) circle (0.7ex);
	
	\draw[darkgray, very thick, dashed] (-6pt,4pt) -- (-6pt,44pt) -- (6pt,44pt) -- (6pt,4pt) -- cycle;
	
	\draw[black, fill=white] (0pt,-10pt) circle (0.7ex);
	\draw[black, fill=white] (0pt,-19pt) circle (0.7ex);
	\draw[black, fill=white] (0pt,-28pt) circle (0.7ex);
	\draw[black, fill=white] (0pt,-37pt) circle (0.7ex);
	
	\draw[darkgray, very thick, dashed] (-6pt,-4pt) -- (-6pt,-44pt) -- (6pt,-44pt) -- (6pt,-4pt) -- cycle;
	
	\draw[black, fill=white] (40pt,23.5pt) circle (0.7ex);
	\draw[black, fill=white] (40pt,-23.5pt) circle (0.7ex);
	
	\draw[gray, thick] (-8pt,-50pt) -- (-8pt,-52pt) -- (8pt,-52pt) -- (8pt,-50pt);
	\draw[gray, thick] (0pt,-52pt) -- (0pt,-54pt);
	\node[anchor=north] at (0pt,-52pt) {\scriptsize groups};
	\node[anchor=north] at (0pt,-59pt) {\scriptsize of $N\!-\!1$};
	\node[anchor=north] at (0pt,-66pt) {\scriptsize nodes};
	
	\draw[gray, thick] (32pt,-50pt) -- (32pt,-52pt) -- (48pt,-52pt) -- (48pt,-50pt);
	\draw[gray, thick] (40pt,-52pt) -- (40pt,-54pt);
	\node[anchor=north] at (40pt,-52pt) {\scriptsize target};
	\node[anchor=north] at (40pt,-59pt) {\scriptsize nodes};


	\draw[darkgray, very thick] (144pt,24pt) -- (144pt,-24pt);
	\draw[black, fill=white] (144pt,24pt) circle (0.8ex);
	\draw[black, fill=white] (144pt,-24pt) circle (0.8ex);
	\node[anchor=north] at (149pt,24pt) {\small $a$};
	\node[anchor=north] at (149pt,-24pt) {\small $b$};
	
	\draw[very thick, gray, arrows=-stealth] (165pt,0pt) -- (205pt,0pt);
	\node[anchor=south] at (185pt,5pt) {\scriptsize \textit{transformed}};
	\node[anchor=south] at (185pt,-1pt) {\scriptsize \textit{to}};

	\draw[midgray, densely dotted, very thick] (228pt,0pt) -- (220pt,0pt) -- (205pt,52pt);
	\node[anchor=south] at (205pt,57pt) {\scriptsize $a$'s contact node for $b$ and};
	\node[anchor=south] at (205pt,49pt) {\scriptsize $b$'s contact node for $a$};
	
	\draw[line width=0.2pt, arrows=-stealth] (230pt,0pt) -- (268pt,21pt);
	\draw[line width=0.2pt, arrows=-stealth] (230pt,19pt) -- (267pt,22.5pt);
	\draw[line width=0.2pt, arrows=-stealth] (230pt,28pt) -- (267pt,24.5pt);
	\draw[line width=0.2pt, arrows=-stealth] (230pt,37pt) -- (268pt,26pt);
	
	\draw[line width=0.2pt, arrows=-stealth] (230pt,0pt) -- (268pt,-21pt);
	\draw[line width=0.2pt, arrows=-stealth] (230pt,-19pt) -- (267pt,-22.5pt);
	\draw[line width=0.2pt, arrows=-stealth] (230pt,-28pt) -- (267pt,-24.5pt);
	\draw[line width=0.2pt, arrows=-stealth] (230pt,-37pt) -- (268pt,-26pt);
	
	\draw[black, fill=white] (230pt,0pt) circle (0.7ex);
	\draw[black, fill=white] (230pt,19pt) circle (0.7ex);
	\draw[black, fill=white] (230pt,28pt) circle (0.7ex);
	\draw[black, fill=white] (230pt,37pt) circle (0.7ex);
	
	\draw[darkgray, very thick, dashed] (223pt,-6pt) -- (223pt,44pt) -- (235pt,44pt) -- (235pt,-6pt) -- cycle;
	
	\draw[black, fill=white] (230pt,-19pt) circle (0.7ex);
	\draw[black, fill=white] (230pt,-28pt) circle (0.7ex);
	\draw[black, fill=white] (230pt,-37pt) circle (0.7ex);
	
	\draw[darkgray, very thick, dashed] (225pt,6pt) -- (225pt,-44pt) -- (237pt,-44pt) -- (237pt,6pt) -- cycle;
	
	\draw[black, fill=white] (270pt,23.5pt) circle (0.7ex);
	\draw[black, fill=white] (270pt,-23.5pt) circle (0.7ex);
	
	\draw[gray, thick] (222pt,-50pt) -- (222pt,-52pt) -- (238pt,-52pt) -- (238pt,-50pt);
	\draw[gray, thick] (230pt,-52pt) -- (230pt,-54pt);
	\node[anchor=north] at (230pt,-52pt) {\scriptsize groups};
	\node[anchor=north] at (230pt,-59pt) {\scriptsize of $N\!-\!1$};
	\node[anchor=north] at (230pt,-66pt) {\scriptsize nodes};
	
	\draw[gray, thick] (262pt,-50pt) -- (262pt,-52pt) -- (278pt,-52pt) -- (278pt,-50pt);
	\draw[gray, thick] (270pt,-52pt) -- (270pt,-54pt);
	\node[anchor=north] at (270pt,-52pt) {\scriptsize target};
	\node[anchor=north] at (270pt,-59pt) {\scriptsize nodes};

\end{tikzpicture}

%% file: vertexcover.tikz
\begin{tikzpicture}
	
	\draw[line width=0.3pt, arrows=-stealth] (77pt,18pt) -- (118.5pt,87.3pt);
	\draw[line width=0.3pt, arrows=-stealth] (77pt,104pt) -- (118.5pt,32.8pt);
	
	
	\draw[line width=0.3pt, arrows=-stealth] (127pt,17pt) -- (176.5pt,17pt);
	\draw[line width=0.3pt, arrows=-stealth] (77pt,13.5pt) -- (117pt,-28pt);
	\draw[line width=0.3pt, arrows=-stealth] (77pt,8pt) -- (117.5pt,-42.5pt);
	
	\draw[black, fill=white] (0pt,0pt) circle (0.8ex);
	\draw[black, fill=white] (0pt,15pt) circle (0.8ex);
	\draw[black, fill=white] (0pt,30pt) circle (0.8ex);
	
	\draw[black, fill=white] (50pt,0pt) circle (0.8ex);
	\draw[black, fill=white] (50pt,15pt) circle (0.8ex);
	\draw[black, fill=white] (50pt,30pt) circle (0.8ex);
	\draw[black, fill=white] (60pt,0pt) circle (0.8ex);
	\draw[black, fill=white] (60pt,15pt) circle (0.8ex);
	\draw[black, fill=white] (60pt,30pt) circle (0.8ex);
	\draw[black, fill=white] (70pt,0pt) circle (0.8ex);
	\draw[black, fill=white] (70pt,15pt) circle (0.8ex);
	\draw[black, fill=white] (70pt,30pt) circle (0.8ex);
	
	\draw[black, fill=white] (120pt,0pt) circle (0.8ex);
	\draw[black, fill=white] (120pt,15pt) circle (0.8ex);
	\draw[black, fill=white] (120pt,30pt) circle (0.8ex);
	
	\draw[black, fill=white] (180pt,17pt) circle (0.8ex);
	
	\draw[black, fill=white] (120pt,-30pt) circle (0.8ex);
	\draw[black, fill=white] (120pt,-45pt) circle (0.8ex);
	
	\draw[darkgray, very thick, dashed] (-7pt,-6pt) -- (77pt,-6pt) -- (77pt,38pt) -- (-7pt,38pt) -- cycle;
	\draw[darkgray, very thick, dashed] (43pt,-8pt) -- (127pt,-8pt) -- (127pt,36pt) -- (43pt,36pt) -- cycle;
	
	\node[anchor=north] at (111pt,32.5pt) {\scriptsize $t_{a,1,b}$};
	\node[anchor=north] at (31pt,-7pt) {\scriptsize $V_{b,1}$};
	\node[anchor=north] at (83pt,-9pt) {\scriptsize $V_{b,2}$};
	\node[anchor=center] at (183.5pt,9.5pt) {\scriptsize $t_{b,2}$};
	
	\draw[lightgray, thick] (125pt,-27pt) -- (127pt,-27pt) -- (127pt,-48pt) -- (125pt,-48pt);
	\draw[lightgray, thick] (127pt,-37.5pt) -- (129pt,-37.5pt);
	\node[anchor=center] at (148pt,-33.5pt) {\scriptsize further target};
	\node[anchor=center] at (148pt,-41.5pt) {\scriptsize nodes of $V_{b,1}$};
	
	
	\draw[line width=0.3pt, arrows=-stealth] (127pt,107pt) -- (176.5pt,107pt);
	\draw[line width=0.3pt, arrows=-stealth] (77pt,108.5pt) -- (117pt,148pt);
	\draw[line width=0.3pt, arrows=-stealth] (77pt,114.5pt) -- (117.5pt,162.5pt);
	
	\draw[black, fill=white] (0pt,90pt) circle (0.8ex);
	\draw[black, fill=white] (0pt,105pt) circle (0.8ex);
	\draw[black, fill=white] (0pt,120pt) circle (0.8ex);
	
	\draw[black, fill=white] (50pt,90pt) circle (0.8ex);
	\draw[black, fill=white] (50pt,105pt) circle (0.8ex);
	\draw[black, fill=white] (50pt,120pt) circle (0.8ex);
	\draw[black, fill=white] (60pt,90pt) circle (0.8ex);
	\draw[black, fill=white] (60pt,105pt) circle (0.8ex);
	\draw[black, fill=white] (60pt,120pt) circle (0.8ex);
	\draw[black, fill=white] (70pt,90pt) circle (0.8ex);
	\draw[black, fill=white] (70pt,105pt) circle (0.8ex);
	\draw[black, fill=white] (70pt,120pt) circle (0.8ex);
	
	\draw[black, fill=white] (120pt,90pt) circle (0.8ex);
	\draw[black, fill=white] (120pt,105pt) circle (0.8ex);
	\draw[black, fill=white] (120pt,120pt) circle (0.8ex);
	
	\draw[black, fill=white] (180pt,107pt) circle (0.8ex);
	
	\draw[black, fill=white] (120pt,150pt) circle (0.8ex);
	\draw[black, fill=white] (120pt,165pt) circle (0.8ex);
	
	\draw[darkgray, very thick, dashed] (-7pt,84pt) -- (77pt,84pt) -- (77pt,128pt) -- (-7pt,128pt) -- cycle;
	\draw[darkgray, very thick, dashed] (43pt,82pt) -- (127pt,82pt) -- (127pt,126pt) -- (43pt,126pt) -- cycle;
	
	\node[anchor=south] at (109pt,86.5pt) {\scriptsize $t_{b,1,a}$};
	\node[anchor=north] at (31pt,83pt) {\scriptsize $V_{a,1}$};
	\node[anchor=north] at (83pt,81pt) {\scriptsize $V_{a,2}$};
	\node[anchor=center] at (183.5pt,99.5pt) {\scriptsize $t_{a,2}$};
	
	\draw[lightgray, thick] (125pt,147pt) -- (127pt,147pt) -- (127pt,168pt) -- (125pt,168pt);
	\draw[lightgray, thick] (127pt,157.5pt) -- (129pt,157.5pt);
	\node[anchor=center] at (148pt,161.5pt) {\scriptsize further target};
	\node[anchor=center] at (148pt,153.5pt) {\scriptsize nodes of $V_{a,1}$};
	
	
	\draw[gray, thick] (111pt,-65pt) -- (111pt,-68pt) -- (129pt,-68pt) -- (129pt,-65pt);
	\draw[gray, thick] (120pt,-68pt) -- (120pt,-71pt);
	\node[anchor=north] at (120pt,-71pt) {\scriptsize $N$ target and};
	\node[anchor=north] at (120pt,-79pt) {\scriptsize extra nodes in};
	\node[anchor=north] at (120pt,-87pt) {\scriptsize each group};
	
	\draw[gray, thick] (-9pt,-65pt) -- (-9pt,-68pt) -- (9pt,-68pt) -- (9pt,-65pt);
	\draw[gray, thick] (0pt,-68pt) -- (0pt,-71pt);
	\node[anchor=north] at (0pt,-71pt) {\scriptsize $N$ extra nodes};
	\node[anchor=north] at (0pt,-79pt) {\scriptsize in each group};
	
	\draw[gray, thick] (42pt,-65pt) -- (42pt,-68pt) -- (78pt,-68pt) -- (78pt,-65pt);
	\draw[gray, thick] (60pt,-68pt) -- (60pt,-71pt);
	\node[anchor=north] at (60pt,-71pt) {\scriptsize $k-N$ common};
	\node[anchor=north] at (60pt,-79pt) {\scriptsize nodes in each};
	\node[anchor=north] at (60pt,-87pt) {\scriptsize group};
	
\end{tikzpicture}

%% file: vc_transform.tikz
\begin{tikzpicture}
	
	\draw[line width=0.3pt, arrows=-stealth] (0pt,95pt) -- (0pt,70pt);
	\draw[line width=0.3pt, arrows=-stealth] (0pt,70pt) -- (-10pt,45pt);
	\draw[line width=0.3pt, arrows=-stealth] (-10pt,45pt) -- (0pt,20pt);
	\draw[line width=0.3pt, arrows=-stealth] (0pt,20pt) -- (0pt,70pt);
	
	\draw[black, fill=white] (0pt,95pt) circle (1ex);
	\draw[black, fill=white] (0pt,70pt) circle (1ex);
	\draw[black, fill=white] (-10pt,45pt) circle (1ex);
	\draw[black, fill=white] (0pt,20pt) circle (1ex);
	
	\draw[lightgray, thick, dotted] (75pt,60pt) -- (75pt,105pt) -- (155pt,105pt) -- (155pt,60pt) -- cycle;
	
	\draw[very thick, gray, arrows=-stealth] (20pt,57.5pt) -- (60pt,57.5pt);
	
	\node[anchor=south] at (40pt,62.5pt) {\scriptsize \textit{transformed}};
	\node[anchor=south] at (40pt,56.5pt) {\scriptsize \textit{to}};
	
	\draw[gray, very thick] (90pt,18.5pt) -- (140pt,18.5pt);
	\draw[gray, very thick] (90pt,21.5pt) -- (140pt,21.5pt);
	\draw[gray, very thick] (90pt,43.5pt) -- (140pt,43.5pt);
	\draw[gray, very thick] (90pt,46.5pt) -- (140pt,46.5pt);
	\draw[gray, very thick] (90pt,68.5pt) -- (140pt,68.5pt);
	\draw[gray, very thick] (90pt,71.5pt) -- (140pt,71.5pt);
	\draw[gray, very thick] (90pt,93.5pt) -- (140pt,93.5pt);
	\draw[gray, very thick] (90pt,96.5pt) -- (140pt,96.5pt);
	
	\draw[arrows=-stealth] (90pt,20pt) -- (140pt,42pt);
	\draw[arrows=-stealth] (90pt,20pt) -- (140pt,66pt);
	\draw[arrows=-stealth] (90pt,45pt) -- (140pt,22pt);
	\draw[arrows=-stealth] (90pt,45pt) -- (140pt,68pt);
	\draw[arrows=-stealth] (90pt,70pt) -- (140pt,24pt);
	\draw[arrows=-stealth] (90pt,70pt) -- (140pt,48pt);
	\draw[arrows=-stealth] (90pt,70pt) -- (140pt,92pt);
	\draw[arrows=-stealth] (90pt,95pt) -- (140pt,73pt);
	
	\draw[fill=white] (80pt,15pt) -- (80pt,25pt) -- (90pt,25pt) -- (90pt,15pt) -- cycle;
	\draw[fill=white] (80pt,40pt) -- (80pt,50pt) -- (90pt,50pt) -- (90pt,40pt) -- cycle;
	\draw[fill=white] (80pt,65pt) -- (80pt,75pt) -- (90pt,75pt) -- (90pt,65pt) -- cycle;
	\draw[fill=white] (80pt,90pt) -- (80pt,100pt) -- (90pt,100pt) -- (90pt,90pt) -- cycle;
	
	\draw[fill=white] (140pt,15pt) -- (140pt,25pt) -- (150pt,25pt) -- (150pt,15pt) -- cycle;
	\draw[fill=white] (140pt,40pt) -- (140pt,50pt) -- (150pt,50pt) -- (150pt,40pt) -- cycle;
	\draw[fill=white] (140pt,65pt) -- (140pt,75pt) -- (150pt,75pt) -- (150pt,65pt) -- cycle;
	\draw[fill=white] (140pt,90pt) -- (140pt,100pt) -- (150pt,100pt) -- (150pt,90pt) -- cycle;
	
	\draw[gray, thick] (76pt,7pt) -- (76pt,5pt) -- (94pt,5pt) -- (94pt,7pt);
	\draw[gray, thick] (85pt,3pt) -- (85pt,5pt);
	\node[anchor=north] at (85pt,5pt) {\scriptsize \textit{first-level}};
	\node[anchor=north] at (85pt,-3pt) {\scriptsize \textit{groups}};
	
	\draw[gray, thick] (136pt,7pt) -- (136pt,5pt) -- (154pt,5pt) -- (154pt,7pt);
	\draw[gray, thick] (145pt,3pt) -- (145pt,5pt);
	\node[anchor=north] at (145pt,5pt) {\scriptsize \textit{second-level}};
	\node[anchor=north] at (145pt,-3pt) {\scriptsize \textit{groups}};

\end{tikzpicture}

%% file: greedy.tikz
\begin{tikzpicture}
	
	\draw[arrows=-stealth] (5pt,10pt) -- (5pt,25pt);
	\draw[arrows=-stealth] (5pt,35pt) -- (5pt,50pt);
	\draw[arrows=-stealth] (5pt,60pt) -- (5pt,75pt);
	
	\draw[arrows=-stealth] (30pt,10pt) -- (30pt,25pt);
	\draw[arrows=-stealth] (30pt,35pt) -- (30pt,50pt);
	
	\draw[arrows=-stealth] (55pt,10pt) -- (55pt,25pt);
	
	\draw[gray, very thick] (31.5pt,6.5pt) -- (6.5pt,31.5pt);
	\draw[gray, very thick] (28.5pt,3.5pt) -- (3.5pt,28.5pt);
	
	\draw[gray, very thick] (56.5pt,6.5pt) -- (6.5pt,56.5pt);
	\draw[gray, very thick] (53.5pt,3.5pt) -- (3.5pt,53.5pt);
	
	\draw[gray, very thick] (81.5pt,6.5pt) -- (6.5pt,81.5pt);
	\draw[gray, very thick] (78.5pt,3.5pt) -- (3.5pt,78.5pt);
	
	\draw[fill=white] (0pt,0pt) -- (0pt,10pt) -- (10pt,10pt) -- (10pt,0pt) -- cycle;
	\draw[fill=white] (25pt,0pt) -- (25pt,10pt) -- (35pt,10pt) -- (35pt,0pt) -- cycle;
	\draw[fill=white] (50pt,0pt) -- (50pt,10pt) -- (60pt,10pt) -- (60pt,0pt) -- cycle;
	\draw[fill=white] (75pt,0pt) -- (75pt,10pt) -- (85pt,10pt) -- (85pt,0pt) -- cycle;
	
	\draw[fill=white] (0pt,25pt) -- (0pt,35pt) -- (10pt,35pt) -- (10pt,25pt) -- cycle;
	\draw[fill=white] (25pt,25pt) -- (25pt,35pt) -- (35pt,35pt) -- (35pt,25pt) -- cycle;
	\draw[fill=white] (50pt,25pt) -- (50pt,35pt) -- (60pt,35pt) -- (60pt,25pt) -- cycle;
	
	\draw[fill=white] (0pt,50pt) -- (0pt,60pt) -- (10pt,60pt) -- (10pt,50pt) -- cycle;
	\draw[fill=white] (25pt,50pt) -- (25pt,60pt) -- (35pt,60pt) -- (35pt,50pt) -- cycle;
	
	\draw[fill=white] (0pt,75pt) -- (0pt,85pt) -- (10pt,85pt) -- (10pt,75pt) -- cycle;
	
	\draw[fill=white, line width=0.25pt] (140pt,10pt) -- (140pt,17pt) -- (147pt,17pt) -- (147pt,10pt) -- cycle;
	\draw[fill=white, line width=0.25pt] (157pt,10pt) -- (157pt,17pt) -- (164pt,17pt) -- (164pt,10pt) -- cycle;
	\draw[fill=white, line width=0.25pt] (174pt,10pt) -- (174pt,17pt) -- (181pt,17pt) -- (181pt,10pt) -- cycle;
	\draw[fill=white, line width=0.25pt] (191pt,10pt) -- (191pt,17pt) -- (198pt,17pt) -- (198pt,10pt) -- cycle;
	
	\draw[fill=white, line width=0.25pt] (140pt,27pt) -- (140pt,34pt) -- (147pt,34pt) -- (147pt,27pt) -- cycle;
	\draw[fill=white, line width=0.25pt] (157pt,27pt) -- (157pt,34pt) -- (164pt,34pt) -- (164pt,27pt) -- cycle;
	\draw[fill=white, line width=0.25pt] (174pt,27pt) -- (174pt,34pt) -- (181pt,34pt) -- (181pt,27pt) -- cycle;
	
	\draw[fill=white, line width=0.25pt] (140pt,44pt) -- (140pt,51pt) -- (147pt,51pt) -- (147pt,44pt) -- cycle;
	\draw[fill=white, line width=0.25pt] (157pt,44pt) -- (157pt,51pt) -- (164pt,51pt) -- (164pt,44pt) -- cycle;
	
	\draw[fill=white, line width=0.25pt] (140pt,61pt) -- (140pt,68pt) -- (147pt,68pt) -- (147pt,61pt) -- cycle;
	
	\node[anchor=south] at (169pt,72pt) {\small \textit{Greedy visiting:}};
	
	\draw[fill=white, line width=0.25pt] (240pt,10pt) -- (240pt,17pt) -- (247pt,17pt) -- (247pt,10pt) -- cycle;
	\draw[fill=white, line width=0.25pt] (257pt,10pt) -- (257pt,17pt) -- (264pt,17pt) -- (264pt,10pt) -- cycle;
	\draw[fill=white, line width=0.25pt] (274pt,10pt) -- (274pt,17pt) -- (281pt,17pt) -- (281pt,10pt) -- cycle;
	\draw[fill=white, line width=0.25pt] (291pt,10pt) -- (291pt,17pt) -- (298pt,17pt) -- (298pt,10pt) -- cycle;
	
	\draw[fill=white, line width=0.25pt] (240pt,27pt) -- (240pt,34pt) -- (247pt,34pt) -- (247pt,27pt) -- cycle;
	\draw[fill=white, line width=0.25pt] (257pt,27pt) -- (257pt,34pt) -- (264pt,34pt) -- (264pt,27pt) -- cycle;
	\draw[fill=white, line width=0.25pt] (274pt,27pt) -- (274pt,34pt) -- (281pt,34pt) -- (281pt,27pt) -- cycle;
	
	\draw[fill=white, line width=0.25pt] (240pt,44pt) -- (240pt,51pt) -- (247pt,51pt) -- (247pt,44pt) -- cycle;
	\draw[fill=white, line width=0.25pt] (257pt,44pt) -- (257pt,51pt) -- (264pt,51pt) -- (264pt,44pt) -- cycle;
	
	\draw[fill=white, line width=0.25pt] (240pt,61pt) -- (240pt,68pt) -- (247pt,68pt) -- (247pt,61pt) -- cycle;
	
	\node[anchor=south] at (269pt,72pt) {\small \textit{Optimal visiting:}};
	
	\draw[gray, fill=gray] (143.5pt,13.5pt) circle (0.2ex);
	\draw[gray, fill=gray] (143.5pt,30.5pt) circle (0.2ex);
	\draw[gray, fill=gray] (143.5pt,47.5pt) circle (0.2ex);
	\draw[gray, fill=gray] (143.5pt,64.5pt) circle (0.2ex);
	\draw[gray, fill=gray] (160.5pt,13.5pt) circle (0.2ex);
	\draw[gray, fill=gray] (160.5pt,30.5pt) circle (0.2ex);
	\draw[gray, fill=gray] (160.5pt,47.5pt) circle (0.2ex);
	\draw[gray, fill=gray] (177.5pt,13.5pt) circle (0.2ex);
	\draw[gray, fill=gray] (177.5pt,30.5pt) circle (0.2ex);
	\draw[gray, fill=gray] (194.5pt,13.5pt) circle (0.2ex);
	
	\draw[gray, fill=gray] (243.5pt,13.5pt) circle (0.2ex);
	\draw[gray, fill=gray] (243.5pt,30.5pt) circle (0.2ex);
	\draw[gray, fill=gray] (243.5pt,47.5pt) circle (0.2ex);
	\draw[gray, fill=gray] (243.5pt,64.5pt) circle (0.2ex);
	\draw[gray, fill=gray] (260.5pt,13.5pt) circle (0.2ex);
	\draw[gray, fill=gray] (260.5pt,30.5pt) circle (0.2ex);
	\draw[gray, fill=gray] (260.5pt,47.5pt) circle (0.2ex);
	\draw[gray, fill=gray] (277.5pt,13.5pt) circle (0.2ex);
	\draw[gray, fill=gray] (277.5pt,30.5pt) circle (0.2ex);
	\draw[gray, fill=gray] (294.5pt,13.5pt) circle (0.2ex);
	
	\draw[gray, arrows=-stealth] (243.5pt,13.5pt) -- (260.5pt,13.5pt) -- (243.5pt,30.5pt) -- (277.5pt,13.5pt) -- (243.5pt,47.5pt)  -- (294.5pt,13.5pt)  -- (244pt,64pt);
	
	\draw[gray, arrows=-stealth] (194.5pt,13.5pt) -- (177.5pt,13.5pt) -- (177.5pt,30.5pt) -- (160.5pt,13.5pt) -- (160.5pt,47.5pt)  -- (143.5pt,13.5pt)  -- (143.5pt,63pt);

\end{tikzpicture}